\documentclass[11pt]{article}
\usepackage[letterpaper,margin=1in]{geometry}

\usepackage[utf8]{inputenc} 
\usepackage[T1]{fontenc}    
\usepackage{url}            
\usepackage{booktabs}       
\usepackage{amsfonts}       
\usepackage{nicefrac}       
\usepackage{microtype}      
\usepackage[table,xcdraw]{xcolor}         
\usepackage{tikz}
\usepackage{authblk}

\usepackage[symbol]{footmisc}

\usepackage{algorithm}
\usepackage[noend]{algpseudocode}
\usepackage{nicefrac,bm}
\algnewcommand{\IIf}[1]{\State\algorithmicif\ #1\ \algorithmicthen}
\algnewcommand{\EndIIf}{\unskip\ \algorithmicend\ \algorithmicif}
\usepackage{amsmath,amsthm,amssymb,mathrsfs,color,colortbl,xspace}

\usepackage{enumitem}
\usepackage{caption,graphicx,subcaption}
\usepackage{natbib}
\usepackage{wrapfig}
\usepackage{comment}
\usepackage{thmtools} 
\usepackage{thm-restate}
\usepackage[hidelinks]{hyperref}
\usepackage[capitalize]{cleveref}

\newcommand{\cA}{\mathcal{A}}
\newcommand{\cB}{\mathcal{B}}

\newcommand{\cE}{\mathcal{E}}
\newcommand{\cF}{\mathcal{F}}

\newcommand{\cH}{\mathcal{H}}
\newcommand{\cI}{\mathcal{I}}

\newcommand{\cR}{\mathcal{R}}
\newcommand{\cS}{\mathcal{S}}

\newcommand{\cX}{\mathcal{X}}
\newcommand{\cY}{\mathcal{Y}}

\newcommand{\bbE}{\mathbb{E}}

\newcommand{\I}{\mathbb{I}}

\newcommand{\N}{\mathbb{N}}

\newcommand{\Pb}{\mathbb{P}}
\newcommand{\Q}{\mathbb{Q}}
\newcommand{\R}{\mathbb{R}}

\newcommand{\e}{\varepsilon}
\newcommand{\eps}{\varepsilon}
\newcommand{\fhi}{\varphi}

\newcommand{\rmd}{\mathrm{d}}
\newcommand{\dif}{\,\rmd}
\newcommand{\diff}{\rmd}

\newcommand{\lrb}[1]{\left(#1\right)}
\newcommand{\brb}[1]{\bigl(#1\bigr)}
\newcommand{\Brb}[1]{\Bigl(#1\Bigr)}

\newcommand{\lsb}[1]{\left[#1\right]}
\newcommand{\bsb}[1]{\bigl[#1\bigr]}
\newcommand{\Bsb}[1]{\Bigl[#1\Bigr]}

\newcommand{\lcb}[1]{\left\{#1\right\}}
\newcommand{\bcb}[1]{\bigl\{#1\bigr\}}

\newcommand{\lce}[1]{\left\lceil#1\right\rceil}
\newcommand{\bce}[1]{\bigl\lceil#1\bigr\rceil}

\newcommand{\labs}[1]{\left\lvert#1\right\rvert}
\newcommand{\babs}[1]{\bigl\lvert#1\bigr\rvert}

\newcommand{\lno}[1]{\left\lVert#1\right\rVert}
\newcommand{\bno}[1]{\bigl\lVert#1\bigr\rVert}

\newcommand{\E}[1]{\mathbb{E}\left[#1\right]}
\newcommand{\Eo}[1]{\mathbb{E}^0\left[#1\right]}
\newcommand{\Ei}[1]{\mathbb{E}^i\left[#1\right]}

\renewcommand{\Pi}[1]{\mathbb{P}^i\left(#1\right)}

\newcommand{\s}{\subset}
\newcommand{\m}{\setminus}
\newcommand{\iop}{\infty}

\newcommand{\ind}[1]{\mathbb{I}{\left\{ #1 \right\}}} 
\DeclareMathOperator{\util}{Util}
\newcommand{\fracc}[2]{#1/#2}
\newcommand{\domain}{\brb{ [0,1] \times \{\star\} } \cup \brb{ \{\star\} \times [0,1] }}
\newcommand{\kl}{\mathrm{KL}}
\newcommand{\tv}{\mathrm{TV}}
\newcommand{\leb}{\mathbb{L}}

\newcommand{\cprob}{\frac{2}{3}}
\newcommand{\cspike}{\frac{1}{144}}
\newcommand{\expthreefpa}{Exp3.FPA}
\newcommand{\expthreefpawrapper}{W.T.FPA}
\newcommand{\expthreefpawrapperLong}{Wrapper for Transparent First-Price Auctions}
\newcommand{\ghat}{\widehat{g}}
\newcommand{\xbar}{\bar{x}}
\newcommand{\jstar}{j^{\star}}

\newcommand{\collect}{\textsc{Collect Bids}\xspace}

\newcommand{\semitransparent}{\textsc{Collecting Bandit}}
\newcommand{\st}{\textsc{CoBa}}

\newcommand{\banditsmooth}{\textsc{Discretized Bandit}}
\newcommand{\cAt}{\widetilde{\cA}}
\newcommand{\tI}{\widetilde{I}}

\renewcommand{\tilde}{\widetilde}
\renewcommand{\hat}{\widehat}

\newcommand{\nsamples}{500}

\newtheorem{theorem}{Theorem}
\newtheorem{proposition}{Proposition}
\newtheorem{lemma}{Lemma}
\newtheorem{claim}{Claim}

\newtheorem{definition}{Definition}
\newtheorem{notation}{Notation}

\title{The Role of Transparency in Repeated First-Price Auctions \\ with Unknown Valuations\thanks{This is the full version of \citet{Cesa-BianchiCRFL24}
}}

\date{}

\author[1,4]{Nicol\`o Cesa-Bianchi}
\author[2]{Tommaso Cesari}
\author[1,4]{Roberto Colomboni}
\author[3]{\\ Federico Fusco}
\author[3]{Stefano Leonardi}

\affil[1]{Universit\`a degli Studi di Milano, Milano, Italy}
\affil[2]{University of Ottawa, Ottawa, Canada}
\affil[3]{Sapienza Universit\`a di Roma, Roma, Italy}
\affil[4]{Politecnico di Milano, Milano, Italy}

\begin{document}

\maketitle
\thispagestyle{empty}

\begin{abstract}
We study the problem of regret minimization for a single bidder in a sequence of first-price auctions where the bidder discovers the item's value only if the auction is won. Our main contribution is a complete characterization, up to logarithmic factors, of the minimax regret in terms of the auction's \emph{transparency}, which controls the amount of information on competing bids disclosed by the auctioneer at the end of each auction. Our results hold under different assumptions (stochastic, adversarial, and their smoothed variants) on the environment generating the bidder's valuations and competing bids. These minimax rates reveal how the interplay between transparency and the nature of the environment affects how fast one can learn to bid optimally in first-price auctions.
\end{abstract}

\clearpage
\tableofcontents

\clearpage
\pagenumbering{arabic}

\section{Introduction}
\label{s:intro}

    The online advertising market has recently transitioned from second to first-price auctions. A remarkable example is Google AdSense's move at the end of 2021 \citep{Wong21}, following the switch made by Google AdManager and AdMob. Earlier examples include OpenX, AppNexus, Index Exchange, and Rubicon  \citep{Sluis17}.
    To increase transparency in first-price auctions, some platforms (like AdManager) have a single bidding session for each available impression (unified bidding) and require all partners to share and receive bid data. After the first-price auction closes, bidders receive the minimum bid price that would have won them the impression \citep{Bigler19}. In practice, advertisers face two main sources of uncertainty in the bidding phase: they ignore the value of the competing bids and, crucially, ignore the actual value of the impression they are bidding on. Indeed, clicks and conversion rates---which are only measured \emph{after} the auction is won and the ad is displayed---can vary wildly over time or highly correlate with competing bids. We remark that ignoring the value of the impression strongly affects the bidder's utility: it may lead to overbidding for an impression of low value or, conversely, underbidding and losing a valuable one. 
    To cope with this uncertainty, advertisers rely on auto-bidders that use the feedback provided in the auctions to learn good bidding strategies. We study the learning problem faced by a single bidder within the framework of regret minimization according to the following protocol:
    \begin{algorithm}[h!]
    \begin{algorithmic}[h!]
    \For{$t=1,2,\ldots, T$}
        \State Valuation $V_t$ and competing bid $M_t$ are privately generated\;
        \State The learner posts a bid $B_t$ and receives utility $\util_t(B_t)$: 
        \[
        \util_t(B_t)=(V_t-B_t) \I \{ B_t \ge M_t \}
        \]
        \State The learner observes some feedback $Z_t$\;
    \EndFor
     \caption*{\textbf{Online Bidding Protocol}}
     \label{a:learning-model}
    \end{algorithmic}
    \end{algorithm}
    
    The bidder has no initial information on the environment and seeks to learn the relevant features of the problem on the fly. The performance of a learning strategy for the bidder---also referred to as the learner---is measured in terms of the difference in total utility with respect to the best fixed bid. This difference is called \emph{regret}, and the main goal is to design strategies with asymptotically vanishing time-averaged regret with respect to the best fixed-bid strategy or, equivalently, regret sublinear in the time horizon.
    
    In this work, we are specifically interested in understanding how the ``transparency'' of the auctions---i.e., the amount of information on competing bids disclosed by the auctioneer \emph{after} the auction takes place---affects the learning process. There is a clear tension regarding transparency: on the one hand, bidders want to receive as much information as possible about the environment to learn the competitor's bidding strategies while revealing as little as possible about their (private) bids. On the other hand, the platform may not want to publicly reveal its revenue (i.e., the winning bid). 
    Our investigation addresses both sides of the ``transparency dilemma''. Our algorithmic results provide bidders with a toolbox of learning strategies to (optimally) exploit the various degrees of transparency, while the tightness of our results fully characterizes the impact of transparency on learnability. This complete picture allows platforms to make an informed decision in choosing their level of transparency, as it is in their interest to create a thriving environment for advertisers.
    
    To model the level of transparency, we distinguish four natural types of feedback $Z_t$, specifying the conditions under which the highest competing bid $M_t$ and the bidder's valuation $V_t$ are revealed to the bidder after each round $t$.
    In the transparent feedback setting, $M_t$ is always observed after the auction is concluded, while $V_t$ is only known if the auction is won, i.e., when $B_t \ge M_t$. In the semi-transparent setting, $M_t$ is only observed when the auction is lost. In other words, in the semi-transparent setting, the platform publicly reveals only the winning bid, whereas in the transparent setting, the platform reveals all bids. 
    We also consider two extreme settings that provide two natural learning benchmarks: full feedback ($M_t$ and $V_t$ are always observed irrespective of the auction's outcome) and bandit feedback ($M_t$ is never observed while $V_t$ is only observed by the winning bidder).
    Note that the learner can compute the value of the utility $\util_t(B_t)$ at time $t$ with any type of feedback, including bandit feedback. In this paper, we characterize the learner's minimax regret not only with respect to the degree of transparency of the auction but also with respect to the nature of the process generating the sequence of pairs $(V_t,M_t)$. In particular, we consider four types of environments: stochastic i.i.d., adversarial, and their smooth
    versions (see \Cref{sec:related} for a discussion about smoothness, and \Cref{s:setting} for the formal definition). 
    
\subsection{Overview of Our Results}
We report here an overview of our results (see also Table~\ref{table:our_results}). For simplicity, we often hide the logarithmic factors with the $\tilde O$ notation.

{
    \renewcommand{\arraystretch}{1.2}
    \begin{table}[t!]
    \centering
    \begin{tabular}{c|cl|cc|}
    \cline{2-5}
    & \multicolumn{2}{c|}{Stochastic i.i.d.} & \multicolumn{2}{c|}{Adversarial} \\ \cline{2-5} 
    & \multicolumn{1}{c|}{Smooth}     & \hfill General \hfill  & \multicolumn{1}{c|}{Smooth} & General     \\ \hline
    \multicolumn{1}{|c|}{Full Feedback}  
        &  \multicolumn{1}{l|}{\cellcolor[HTML]{b0ff98}  Thm.\ref{thm:lower-iid-smooth-full}: $\Omega(\sqrt{T})$} 
        & \cellcolor[HTML]{b0ff98} 
        & \multicolumn{1}{l|}{ \cellcolor[HTML]{b0ff98} 
        } 
        & \cellcolor[HTML]{ffa7a4}  Thm.\ref{thm:lower-adv-general-full}: $\Omega(T)$
        \\ \hline
    \multicolumn{1}{|c|}{Transparent}      
        & \multicolumn{1}{l|}{\cellcolor[HTML]{b0ff98} 
        }     
        & \cellcolor[HTML]{b0ff98} 
 Thm.\ref{thm:upper-iid-general-transparent}: $O(\sqrt{T})$  
        & \multicolumn{1}{l|}{\cellcolor[HTML]{b0ff98}  Thm.\ref{thm:upper-adv-smooth-transparent}: $\tilde O(\sqrt{T})$} 
        & \cellcolor[HTML]{ffa7a4} 
        \\ \hline
    \multicolumn{1}{|c|}{Semi-Transparent} 
        & \multicolumn{1}{l|}{\cellcolor[HTML]{fff076}\hspace{-2.5pt}  Thm.\ref{thm:lower-iid-smooth-semi-transparent}: $\Omega\big(T^{\nicefrac 23}\big)$} 
        & \cellcolor[HTML]{fff076}  Thm.\ref{thm:upper-iid-general-semi-transparent}: $\tilde  O\big(T^{\nicefrac 23}\big)$
        & \multicolumn{1}{l|}{\cellcolor[HTML]{fff076} 
        } 
        & \cellcolor[HTML]{ffa7a4} 
        \\ \hline
    \multicolumn{1}{|c|}{Bandit Feedback}  
        & \multicolumn{1}{l|}{\cellcolor[HTML]{fff076} 
        } 
        & \cellcolor[HTML]{ffa7a4}  Thm.\ref{thm:lower-iid-general-bandit}: $\Omega(T)$  & \multicolumn{1}{c|}{\cellcolor[HTML]{fff076}  Thm.\ref{thm:upper-adv-smooth-bandit}: $O\big(T^{\nicefrac 23}\big)$} 
        & \cellcolor[HTML]{ffa7a4} 
        \\ \hline
    \end{tabular}
    \caption{Summary of our results. Rows correspond to feedback models while columns to environments. 
    The minimax regret of every problem falls in one of the following three regimes: $\tilde{\Theta}(\sqrt T)$ (green), $\tilde{\Theta}(T^{\nicefrac 23})$ (yellow) and $\tilde{\Theta}(T)$ (red). }
    \label{table:our_results}
    \end{table}
    }

\paragraph{Stochastic i.i.d.~settings} 
    \begin{itemize}[topsep=2pt,itemsep=0pt,leftmargin=2.5ex,parsep=1pt]
        \item In both the full and transparent feedback models, the minimax regret is of order $\sqrt T$ (\Cref{thm:lower-iid-smooth-full,thm:upper-iid-general-transparent}), and adding the smoothness requirement leaves this rate unchanged.
        \item In the semi-transparent feedback model, the minimax regret is of order $T^{\nicefrac 23}$ (\Cref{thm:lower-iid-smooth-semi-transparent,thm:upper-iid-general-semi-transparent}). Also in this case, adding the smoothness requirement leaves this rate unchanged.
        \item In the bandit feedback model, smoothness is crucial for sublinear regret (\Cref{thm:lower-iid-general-bandit}). In particular, smoothness implies a minimax regret of $T^{\nicefrac 23}$ (this is obtained by combining the upper bound in \Cref{thm:upper-adv-smooth-bandit} and the lower bound in \Cref{thm:lower-iid-smooth-semi-transparent}). 
    \end{itemize}
    \paragraph{Adversarial settings} 
    \begin{itemize}[topsep=2pt,itemsep=0pt,leftmargin=2.5ex,parsep=1pt]
        \item Without smoothness, sublinear regret cannot be achieved, even with full feedback (\Cref{thm:lower-adv-general-full}).
        \item In both the full and transparent feedback model, the minimax regret in a smooth environment is of order $\sqrt T$ (combining the lower bound in \Cref{thm:lower-iid-smooth-full} and the upper bound in \Cref{thm:upper-adv-smooth-transparent}).
        \item Both with semi-transparent and bandit feedback, the minimax regret in a smooth environment is of order $T^{\nicefrac 23}$ (combining the lower bound in \Cref{thm:lower-iid-smooth-semi-transparent} and the upper bound in \Cref{thm:upper-adv-smooth-bandit}). 
    \end{itemize}

    Interestingly, the minimax regret rates for first-price auctions mirror the allowed regret regimes in finite partial monitoring games \citep{BartokFPRS14} and online learning with feedback graphs \citep{AlonCGMMS17}. This is somehow surprising, as it has been shown in  \citet{lattimore2022minimax} that games with continuous outcome/action spaces allow for a much larger set of regret rates---see also \citet{cesa23COLT,
    cesa2024JMLR, 
bolic2023online,BernasconiCCF24}. 
    
    Table \ref{table:our_results} reveals some interesting properties of the learnability of the problem: full feedback and transparent feedback are essentially equivalent, while semi-transparent feedback and bandit feedback differ only in the stochastic i.i.d.\ setting. Qualitatively, this tells the platform that disclosing all bids (instead of only the winning one) drastically improves the learnability of the problem (green vs.\ yellow entries in Table~\ref{table:our_results}). Besides, revealing at least the winning bid avoids some pathological behavior (yellow entries vs.\ red entry for the general i.i.d.\ environment with bandit feedback). Moreover, while smoothness is key for learning in the adversarial setting, in the stochastic case smoothness is only relevant for bandit feedback.
    \begin{figure}[ht]
        \centering
        \begin{tikzpicture}[scale=3]
            \draw[->] (0,-0.1) -- (0, 0.5);
            \draw[->] (-0.1, 0) -- (1.1, 0);
            \draw 
                (0.5, -0.02) -- (0.5, 0.02)
                (0.5, 0) node[below] {\footnotesize $M_t$}
                (0.9, -0.02) -- (0.9, 0.02)
                (0.9, 0) node[below] {\footnotesize $V_t$}
            ;
            \draw[thick, blue] 
                (0,0) -- (0.5,0)
                (0.5,{0.9-0.5}) -- (1,{0.9-1})
            ;
            \draw[thick, blue, fill=white] (0.5,0) circle (0.3pt);
            \draw[thick, blue, fill=blue] (0.5,{0.9-0.5}) circle (0.3pt);
            \draw 
                (1.1,0) node[below] {\footnotesize bid}
                (0,0.5) node[left] {\footnotesize $\util_t$}
            ;
                    \draw[->] (0,{0.5 + 0.9 - 2 - 0.3}) -- (0, {0.1- 0.3});
            \draw[->] (-0.1, - 0.3) -- (1.1, - 0.3);
            \draw 
                (0.65, {-0.02- 0.3}) -- (0.65, {0.02- 0.3})
                (0.65, - 0.3) node[below] {\footnotesize $M_t$}
                (0.5, {-0.02- 0.3}) -- (0.5, {0.02- 0.3})
                (0.5, - 0.3) node[below] {\footnotesize $V_t$}
            ;
            \draw[thick, blue] 
                (0,- 0.3) -- (0.65,- 0.3)
                (0.65,{0.5-0.65- 0.3}) -- (1,{0.5-1- 0.3})
            ;
            \draw[thick, blue, fill=white] (0.65,- 0.3) circle (0.3pt);
            \draw[thick, blue, fill=blue] (0.65,{0.5-0.65- 0.3}) circle (0.3pt);
            \draw 
                (1.1,- 0.3) node[below] {\footnotesize bid}
                (0,{0.1- 0.3}) node[left] {\footnotesize $\util_t$}
            ;
        \end{tikzpicture}
        \hspace{30pt}
        \begin{tikzpicture}[scale=3]
            \draw[->] (0,{0.5 + 0.9 - 2 - 0.15}) -- (0, {0.5 + 0.15});
            \draw[->] (-0.1, 0) -- (1.1, 0);
            \draw[thick, blue] 
                (0,0) -- (0.5,0)
                (0.5,{0.9-0.5}) -- (0.65,{0.9 - 0.65})
                (0.65,{0.9-0.65+0.5-0.65}) -- (1,{0.5 + 0.9 - 2})
            ;
            \draw[thick, blue, fill=white] (0.5,0) circle (0.3pt);
            \draw[thick, blue, fill=blue] (0.5,{0.9-0.5}) circle (0.3pt);
            \draw[thick, blue, fill=white] (0.65,{0.9 - 0.65}) circle (0.3pt);
            \draw[thick, blue, fill=blue] (0.65,{0.9-0.65+0.5-0.65}) circle (0.3pt);
            \draw 
                (1.1,0) node[below] {\footnotesize bid}
                (0,{0.5 + 0.15}) node[left] {\footnotesize $\util$}
            ;
            \draw[thick, red] (0.5,-0.15) -- (0.65,-0.15);
            \draw (0.575,0) node[below] {\footnotesize \textcolor{red}{$\Delta$}};
            \draw (0.52,0) node[above] {\footnotesize \textcolor{red}{$b^\star$}};
            \draw[thick, red, fill=white] (0.65,-0.15) circle (0.3pt);
            \draw[thick, red, fill=red] (0.5,-0.15) circle (0.3pt);
        \end{tikzpicture}
        \caption{The utility function is generally neither Lipschitz nor continuous. If $M_t \le V_t$ (top left plot), then $\util_t$ is upper-semi continuous and one-sided Lipschitz; conversely, if $M_t \ge V_t$ (bottom left plot), then $\util_t$ is still one-sided Lipschitz---from the other side---and lower-semi continuous. Summing up the two types of utilities results in a total utility that may be neither one-sided Lipschitz nor semi-continuous (right plot, where the two utility functions of the other two plots are summed up. There, $b^\star$ is the optimal bid and $\Delta$ is the neighborhood of $b^\star$ where the total utility is ``good enough'').}
        \label{fig:non-lip-utilities}
    \end{figure}
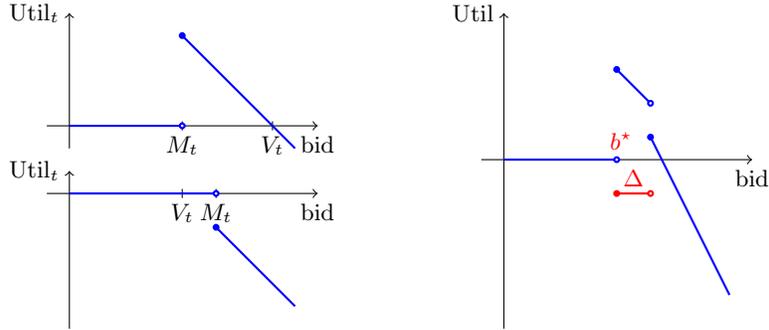

\subsection{Technical Challenges} 

    \paragraph{The utility function.} The utilities $\util_t(b)=(V_t-b)\ind{M_t \le b}$ are defined over a continuous decision space $[0,1]$ and are neither Lipschitz nor continuous, see \Cref{fig:non-lip-utilities}. Actually, even weaker properties, i.e., that the expected cumulative reward $b\mapsto \sum_{t\in[T]}\bbE\bsb{\util_t(b)}$ is one-sided Lipschitz or semi-continuous, do not hold in general. We address this problem by developing techniques designed to control the approximation error incurred when discretizing the bidding space. This is a non-trivial problem without regularity assumption, as the neighborhood of the optimal bid where the total utility is ``good enough'' can be arbitrarily small in general (see the red interval $\Delta$ in the rightmost plot of \Cref{fig:non-lip-utilities}).  
    In the stochastic i.i.d.\ setting, the approximation error is controlled by building a sample-based \emph{non-uniform} grid of candidate bids, which can be of independent interest. This allows us to estimate the distribution of the competing bids uniformly over the subintervals of $[0,1]$. In the adversarial setting, instead, we use the smoothness assumption to guarantee that the expected utility is Lipschitz. In this case, the approximation error is controlled using a uniform grid with an appropriate grid-size (\Cref{lem:lip}).

    \paragraph{The feedback models.} Our feedback models interpolate between bandit (only the bidder's utility is observed) and full feedback ($V_t$ and $M_t$ are always observed).
    In the stochastic i.i.d.\ case, the different levels of transparency are crucial to the process of building the non-uniform grids used to control the discretization error.
    In the adversarial case, when there are only $K$ allowed bids, the optimal rates are of order $\sqrt{T \ln K}$ and $\sqrt{KT}$ under full and bandit feedback, respectively. 
    While the semi-transparent feedback is not enough to improve on the bandit rate, the transparent one can be exploited via a more sophisticated approach. To this end, we design an algorithm, \expthreefpa{}, enjoying the full feedback regret rate of order $\sqrt{T \ln K}$ while only relying on the weaker transparent feedback.
    \paragraph{Lower bounds.} The linear lower bounds (\Cref{thm:lower-adv-general-full,thm:lower-iid-general-bandit}) exploit a ``needle in a haystack'' phenomenon, where there is a hidden optimal bid $b^\star$ in the $[0,1]$ interval and the learner has no way of finding $b^\star$ using the feedback it has access to. This is indeed the case in the non-smooth adversarial full-feedback setting and in the non-smooth i.i.d.\ bandit setting. To prove the remaining lower bounds, we design careful embeddings of known hard instances into our framework. In particular, in \Cref{thm:lower-iid-smooth-full} we embed the hard instance for prediction with two experts and in \Cref{thm:lower-iid-smooth-semi-transparent} the hard instance for $K= \Theta(T^{\nicefrac 13})$ bandits. 

\subsection{Related Work}
\label{sec:related}

\paragraph{Transparency in first-price auctions.} The role of transparency in repeated first-price auctions has been investigated by \citet{BergemannH18}, but mostly from a game-theoretic viewpoint. In particular, they study the impact of the feedback policy on the bidders' strategy and show how disclosing the bids at the end of each round affects the equilibria of a bidding game with infinite horizon. In contrast, we want to characterize the impact of different amounts of feedback (or degrees of transparency) on the learner's regret, which is measured against the optimal fixed bid in hindsight. 
%

\paragraph{Auctions with unknown valuations.} Although the problem of regret minimization in first-price auctions has been studied before, only a few papers consider the natural setting of unknown valuations.
\citet{FengPS18} introduce a general framework for the study of regret in auctions where a bidder's valuation is only observed when the auction is won. In the special case of first-price auctions, their setting is equivalent to our transparent feedback when the sequence of pairs $(V_t,M_t)$ is adversarially generated. Following a parameterization introduced by \citet{WeedPR16}, \citet{FengPS18} provide a $O\big(\sqrt{T\ln\max\{\Delta_0^{-1},T\}}\big)$ regret bound, where $\Delta_0 = \min_{t < t'}|M_t-M_{t'}|$ is controlled by the environment. In the stochastic i.i.d.\ case, their results translate into \emph{distribution-dependent} guarantees that do not translate into a worst-case sublinear bound (we obtain a $\sqrt{T}$ rate). In the adversarial case, their guarantees are still linear in the worst-case (we obtain $\sqrt{T}$  bounds by leveraging the smoothness assumption).
\citet{AchddouCG21} consider a stochastic i.i.d.\ setting with the additional assumption that $V_t$ and $M_t$ are independent. Their main result is a bidding algorithm with \emph{distribution-dependent} regret rates (of order $T^{\nicefrac 13 + \varepsilon}$ or $\sqrt{T}$, depending on the assumptions on the underlying distribution) in the transparent setting. Again, this result is not comparable to ours because of the independence assumption and the distribution-dependent rates (which do not allow to recover our minimax rates). Other works consider regret minimization in repeated second-price auctions with unknown valuations. \citet{dikkala2013can} investigate a repeated bidding setting, but do not consider regret minimization. \citet{WeedPR16} derive regret bounds for the case when $M_t$ are adversarially generated, while $V_t$ are stochastically or adversarially generated and the feedback is transparent.

\paragraph{First-price auctions with known valuations.}
Considerably more works study first price auctions when the valuation $V_t$ is known to the bidder at the beginning of each round $t$. Note that these results are not directly comparable to ours.
\citet{balseiro2019contextual} look at the case when the $V_t$ are adversarial and the $M_t$ are either stochastic i.i.d.\ or adversarial. In the bandit feedback case (when $M_t$ is never observed), they show that the minimax regret is $\widetilde{\Theta}\big(T^{\nicefrac 23}\big)$ in the stochastic case and $\widetilde{\Theta}\big(T^{\nicefrac 34}\big)$ in the adversarial case.
\citet{han2020optimal} prove a $\widetilde{O}\big(\sqrt{T}\big)$ regret bound in the semi-transparent setting ($M_t$ observed only when the auction is lost) with adversarial valuations and stochastic bids.
\citet{han2020learning} focus on the adversarial case, when $V_t$ and $M_t$ are both generated adversarially. They prove a $\widetilde{O}\big(\sqrt{T}\big)$ regret bound in the full feedback setting ($M_t$ always observed) when the regret is defined with respect to all Lipschitz shading policies. 
This setup is extended in \citet{zhang2022leveraging} where the authors consider the case in which the bidder is provided access to hints before each auction. \citet{zhang2021meow} also studied the full information feedback setting and design a space-efficient variant of the algorithm proposed by \citet{han2020learning}.
\citet{badanidiyuru2023learning} introduce a contextual model in which $V_t$ is adversarial and $M_t = \langle \theta,x_t\rangle + \varepsilon_t$ where $x_t\in\R^d$ is contextual information available at the beginning of each round $t$, $\theta\in\R^d$ is an unknown parameter, and $\varepsilon_t$ is drawn from an unknown log-concave distribution. They study regret in bandit and full feedback settings.

\paragraph{Dynamics in first-price auctions.}
A different thread of research is concerned with the convergence property of the regret minimization dynamics in first-price auctions (or, more specifically, with the learning dynamics of mean-based regret minimization algorithms).
\citet{feldman2016correlated} show that with continuous bid levels, coarse-correlated equilibria exist whose revenue is below the second price.
\citet{feng2021convergence} prove that regret minimizing bidders converge to a Bayesian Nash equilibrium in a first-price auctions when bidder values are drawn i.i.d.\ from a uniform distribution on $[0,1]$.  \citet{kolumbus2022auctions} show that if two bidders with finitely many bid values converge, then the equilibrium revenue of the bidder with the highest valuation is the second price.
\citet{deng2022nash} characterize the equilibria of the learning dynamics depending on the number of bidders with the highest valuation. Their characterization is for both time-average and last-iterate convergence.

\paragraph{Smoothed adversary.} Smoothed analysis of algorithms, originally introduced by \citet{spielman2004smoothed} and later formalized for online learning by \citet{rakhlin2011online,haghtalab2020smoothed}, is a known approach to the analysis of algorithms in which the instances at every round are generated from a distribution that is not too concentrated. Recent works on the smoothed analysis of online learning algorithms include
\citet{kannan2018smoothed,haghtalab2020smoothed,haghtalaboracle,block2022smoothed,DurvasulaHZ23,cesa23COLT,cesa2024JMLR,cesa21EC,cesa23MOR,bolic2023online}.

\paragraph{Online learning in metric spaces.}
Our problem is related to online learning in metric spaces \citep{KleinbergSU19}, where the action space is endowed with a metric and the losses are induced by a sequence of Lipschitz functions defined onto it. Tight regret bounds are known, parameterized by some notion of dimension of the metric space, in both the full and the bandit models. The simple structure of our action space ($[0,1]$ with the Euclidean distance) allows us to obtain tight bounds by either using a uniform grid (\Cref{thm:upper-adv-smooth-transparent,thm:upper-adv-smooth-bandit}) or sample-based grids (\Cref{thm:upper-iid-general-semi-transparent,thm:upper-iid-general-transparent}), without resorting to the more elaborate techniques that characterize this line of research, e.g., zooming (which is typically used in the bandit feedback model to account for the lack of feedback).
Also related to our model is the study of piecewise and regular Lipschitz functions \citep{BalcanDV18,SharmaBD20,DuettingGSW23}. In particular, Lemma 1  and Theorem 3 in \citet{BalcanDV18} imply our \Cref{thm:upper-adv-smooth-bandit} in the special case of independent processes.\footnote{Combining the second part of their Lemma 1 with their Theorem 3 to lift independence gives void guarantees in the general case (note that there is a typo in the statement of their Lemma 1: as it can be seen in the proof, the correct result is $k = P \cdot \mathcal O\big( M\cdot \kappa \cdot w + \sqrt{M\log(\nicefrac P\zeta)}\big)$) and, without assuming independence, $P=T$ in our setting).}

\section{The Learning Model}
\label{s:setting}

    We introduce formally the repeated bidding problem in first-price auctions. 
    At each time step $t$, a new item arrives for sale, for which the learner holds some unknown valuation $V_t\in [0,1]$.
    The learner bids some $B_t \in [0,1]$ and, at the same time, a set of competitors bid for the same object. We denote their highest competing bid by $M_t\in [0,1]$. 
    The learner gets the item at cost $B_t$ if it wins the auction (i.e., if $B_t \ge M_t$), and does not get it otherwise. Then, the learner observes some feedback $Z_t$ and gains utility $\util_t(B_t)$, where, for all $b\in[0,1]$,
    $
        \util_t(b) = ( V_t - b ) \I \{ b \ge M_t \}
    $ (see the Protocol in \Cref{s:intro}).
    Crucially, at time $t$ the learner does not know its valuation $V_t$ for the item before bidding, implying that its bid $B_t$ only depends on its past observations $Z_1,\dots,Z_{t-1}$ (and, possibly, some internal randomization).
    The goal of the learner is to design a learning algorithm $\cA$ that maximizes its utility. More precisely, we measure the performance of an algorithm $\cA$ by its \emph{regret} $R_T(\cA)$ against the worst environment $\cS$ in a certain class $\Xi$: $R_T(\cA) = \sup_{\cS \in \Xi}  R_T (\cA,\cS)$, where
    \[
        R_T (\cA,\cS) = 
        \sup_{b \in [0,1]} \E{\sum_{t=1}^T \util_t(b) - \sum_{t=1}^T \util_t(B_t)}.
    \]
    The expectation in the previous display is taken with respect to the randomness of the algorithm $\cA$ which selects $B_t$, and (possibly) the randomness of the environment $\cS$ generating the $(V_t,M_t)$ pairs.
    \paragraph{The environments.} In this paper we consider both stochastic i.i.d.\ and adversarial environments.
    \begin{itemize}[topsep=2pt,itemsep=0pt,leftmargin=2.5ex,parsep=1pt]
        \item Stochastic i.i.d.: The pairs $(V_1,M_1),(V_2,M_2),\dots$ are a stochastic i.i.d.\ process.
        \item Adversarial: The sequence $(V_1,M_1),(V_2,M_2),\dots$ is generated by an oblivious adversary.
    \end{itemize}

    Following previous works in online learning (see \Cref{sec:related}), we also study versions of the above environments that are constrained to generate the sequence of $(V_t,M_t)$ values using distributions that are ``not too concentrated''. To this end, we introduce the notion of smooth distributions. 

    \begin{definition}[\citet{HaghtalabRS21}]
        Let $\cX$ be a domain that supports a uniform distribution $\nu$. A measure $\mu$ on $\cX$ is said to be $\sigma$-smooth if for all measurable subsets $A \subseteq \cX$, we have $\mu(A) \le \frac{\nu(A)}{\sigma}$.
    \end{definition}

    We thus also consider the following two types of environments.
    \begin{itemize}[topsep=2pt,itemsep=0pt,leftmargin=2.5ex,parsep=1pt]
        \item The $\sigma$-smooth stochastic i.i.d.\ environment, which is a stochastic i.i.d.\ environment where the common distribution of all pairs $(V_1,M_1)$, $(V_2,M_2),\dots$ is $\sigma$-smooth.
        \item The $\sigma$-smooth adversarial setting, where the pairs $(V_1,M_1)$, $\dots$ form a stochastic process such that, for each $t$, the distribution of the pair $(V_t,M_t)$ is $\sigma$-smooth.
    \end{itemize}


    \paragraph{The feedback.} After describing the environments that we study, we now specify the types of feedback the learner receives at the end of each round, from the richest to the least informative. 
    \begin{itemize}[topsep=2pt,itemsep=0pt,leftmargin=2.5ex,parsep=1pt]
        \item{Full feedback}. The learner observes its valuation and the highest competing bid: $Z_t=(V_t,M_t)$. 
        \item Transparent feedback. The learner always observes $M_t$, but $V_t$ is only revealed if it gets the item: $Z_t$ is equal to $(\star,M_t)$ if $B_t < M_t$ and to $(V_t,M_t)$ otherwise.
        \item Semi-transparent feedback\footnote{This feedback is similar to the winner-only feedback in \citet{han2020optimal}.}. 
        The learner observes $V_t$ if it gets the item and $M_t$ otherwise: $Z_t$ is equal to $(\star,M_t)$ if $B_t < M_t$ and to $(V_t,\star)$ otherwise.
        \item The bandit feedback\footnote{We call this the bandit feedback because it is equivalent to receiving $\util_t(B_t)$ (with the extra information $\star$ to distinguish between losing the item and winning it with $V_t = B_t$, which does not affect regret guarantees).}. 
        The learner observes $V_t$ if it gets the item and the symbol $\star$ otherwise: $Z_t$ is $\star$ if $B_t < M_t$ and to $V_t$ otherwise.
    \end{itemize}



\section{The Stochastic i.i.d.\ Setting}

    In this section, we investigate the problem of repeated bidding in first-price auctions with unknown valuations, when the pairs of valuations and highest competing bids are drawn i.i.d.\ from a fixed but unknown distribution. We start by proving in \Cref{sec:iid-bandit} that it is impossible to achieve sublinear regret under the bandit feedback model without any assumption on the distribution of the environment. 
    Then, in \Cref{sec:iid-semi}, we give matching upper and lower bounds of order $T^{\nicefrac 23}$ in the semi-transparent feedback model. 
    Notably, the lower bound holds for smooth distributions, while the upper bound works for any (possibly non-smooth) distributions. 
    Finally, in \Cref{sec:iid-semi-full} we prove
    that both the full and transparent feedback yield the same minimax regret regime of order $\sqrt{T}$, regardless of the regularity of the distribution.

    \subsection{I.I.D.\ -- Bandit Feedback}
    \label{sec:iid-bandit}

        In the bandit feedback model, at each time step, the learner observes the valuation $V_t$ (and nothing else) when the auction is won, and observes nothing when the auction is lost. 
        The crucial difference with the other (richer) types of feedback is the amount of information received about $M_t$, which, in the bandit case, is just the relative position with respect to $B_t$ (i.e., whether $M_t \le B_t$ or $B_t < M_t$). 
        This allows to hide in the interval $[0,1]$ an optimal bid $b^\star$ which the learner cannot uncover over a finite time horizon. Following this idea, a difficult environment should randomize between two scenarios: a good scenario with large value $V_t = 1$ and $M_t$ slightly smaller than $b^\star$ and a bad one with poor value $V_t = 0$ and $M_t$ slightly larger than $b^\star$. Then, to avoid suffering linear regret, the learner has to find this tiny interval around $b^\star$ (the ``needle in a haystack'').
        
        \begin{theorem}
        \label{thm:lower-iid-general-bandit}
            Consider the problem of repeated bidding in first-price auctions in a stochastic i.i.d.\ environment with bandit feedback. Then, any learning algorithm $\cA$ satisfies $
                R_T(\cA)
            \ge
                \tfrac 1{13} T $.
        \end{theorem}
        \begin{proof}
            We construct a randomized i.i.d.\ environment $\cS$, such that any deterministic algorithm $\cA$ suffers linear regret against it, and then apply Yao's minimax principle to conclude the proof. The randomized environment is simple: before starting the sequence, a uniform seed $b^\star$ is drawn uniformly at random in $(\nicefrac 13, \nicefrac 12-\e)$, where $\e$ is a small parameter we set later. Then, the i.i.d.\ sequence $(V_1,M_1), (V_2,M_2), \dots$ is drawn as follows: at each time step $t$ with probability $\nicefrac 12$ we have $(V_t,M_t) =  \left(1,b^\star\right)$, otherwise $\left(0,b^\star+\e\right)$. The best bid in hindsight, $b^\star$, yields an overall expected utility of $\tfrac T2 (1-b^\star)$, which is at least $\nicefrac T4$, as $b^\star$ belongs to $(\nicefrac 13, \nicefrac 12)$.
            
            We now upper bound the utility achievable by any deterministic algorithm $\cA$ against $\cS$. Fix any such algorithm, and consider its bids against any environment that selects the valuations $V_t$ to be either $0$ or $1$ (as the one we just constructed). At each time step, the feedback that $\cA$ receives is $0$, $1$ or $\star$ (when the item is allocated to one of the competitors), so that the history of the bids posted by $\cA$ is naturally described by a ternary decision tree of height $T$, where each level corresponds to a time step and any node to a bid. Crucially, the leaves of this tree are finite (at most $3^T$), which means that the algorithm $\cA$ only posts bids in a finite subset $N$ of $[0,1]$. Now, let $\e = \nicefrac{3^{-{2T}}}{12}$; we have that, with probability at least $1-\nicefrac{6N\e}{(1-6\e)} \ge 1 - e^{-T}$, the set $[b^\star, b^\star+\e]$ does not intersect $N$. Note: the randomness is with respect to the uniform seed $b^\star$ drawn by $\cS$, while the bound on the probability holds independently to the choice of the deterministic algorithm $\cA$.

            The total utility of $\cA$ when $[b^\star, b^\star+\e]$ does not intersect $N$ is easy to analyze: every time that $\cA$ posts bids smaller than $b^\star$, then it never wins the item (zero utility). Instead, if it posts bids larger than $b^\star+\e$, then it always gets the item (whose average value is $\nicefrac 1{2}$), paying at least $b^\star+\e \ge \nicefrac 13$. Putting these two cases together, we have proved that at each time step the expected utility earned by the learner is at most $\nicefrac 16 = \nicefrac 12 - \nicefrac 13$, when $[b^\star, b^\star+\e] \cap N = \varnothing$ (which happens with probability at least $1 - e^{-T}$). Finally, by combining the lower bound on the performance of $b^\star$ with the upper bound on the expected utility of the learner, we get
            $
                R_T(\cA, \cS) \ge (1 - e^{-T})(\nicefrac T4 - \nicefrac T6) \ge \nicefrac T{13}. 
            $
        \end{proof}

    \subsection{I.I.D.\ -- Semi-Transparent Feedback}\label{sec:iid-semi}

        In this section, we prove two results settling the minimax regret for the semi-transparent feedback where the environment is i.i.d.\ (and, possibly, smooth). First, we construct a learning algorithm, \semitransparent, achieving $T^{\nicefrac 23}$ regret against any i.i.d.\ environment. Then, we complement it with a lower bound of the same order (up to log terms) obtained even in a smooth i.i.d.\ environment.

        \subsubsection[\texorpdfstring{A $T^{\nicefrac 23}$ upper bound for the i.i.d.\ environment}{A T\^23 upper bound of the i.i.d.\ environment}]{A $T^{\nicefrac 23}$ Upper Bound for the i.i.d.\ Environment}

        Our learning algorithm \semitransparent\ is composed of two phases. First, for $T_0=\Theta(T^{\nicefrac 23})$ rounds, it collects samples from the highest competing bid random variables $M_1,M_2,\dots, M_{T_0}$ by posting dummy bids $B_1=B_2=\dots=B_{T_0}=0$. Among these values (plus the value $X_0 = 0$), the algorithm selects $\Theta(\sqrt{T_0})$ candidate bids according to their ordering, in such a way that the empirical frequencies of bids $M_1,M_2,\dots,M_{T_0}$ landing strictly in between two consecutive selected values are at most $\Theta(\nicefrac1{\sqrt{T_0}})$ (see the pseudocode of \collect for details). Second, for the remaining time steps, it runs any bandit algorithm, using as candidate bids the ones collected in the first phase (see \semitransparent\ for details). Note that, in this second phase, the (less informative) bandit feedback would be enough to run the algorithm: the additional information provided by the semi-transparent feedback is only exploited in the initial ``collecting bids'' phase. 

            \begin{algorithm}[t!]
            \caption*{\collect}\label{alg:collect_bids}
            \begin{algorithmic}[1]
                \State \textbf{input:} Time horizon $T_0$
                \State $X_0 \gets 0$ and $M^{(0)}\gets 0$
                \For{each round $t = 1, 2, \dots, T_0$}
                    \State Post bid $B_t = 0$ and observe the highest competing bid $M_t$
                \EndFor
                \State Sort the observed highest competing bids in increasing order: $M^{(1)} \le M^{(2)} \le \dots \le M^{(T_0)}$
                \State \algorithmicif{} $M^{(T_0)}=0$ \algorithmicthen{} \textbf{return} candidate bid $X_0$
                \For{$i=1, 2, \dots$}
                    \State $\jstar_{i-1} \gets \max \bcb{ j\in\{0,\dots,T_0\} \mid X_{i-1} = M^{(j)} }$
                    \State
                    $j_i \gets \min\bcb{ \jstar_{i-1} + \lce{ \sqrt{T_0} }, T_0 }$,
                    $X_i \gets M^{(j_i)}$
                    \State \algorithmicif{} $j_i = T_0$ \algorithmicthen{} let $K \gets i$ and \textbf{break};
                \EndFor
                \State \textbf{return} Candidate bids $X_0, X_1, X_2, \dots, X_K$
            \end{algorithmic}
            \end{algorithm}
            
            As a first step, we state a simple concentration result pertaining the i.i.d.\ process $M,M_1,M_2,\dots,M_{T_0}$, for $T_0 \in \N$. If $\cI$ is the family of all the subintervals of $[0,1]$ and
            $\delta\in(0,1)$, we define
            \[
                \cE_\delta^{T_0} = \bigcap_{I \in \cI} \textstyle{\left\{ \labs{ \frac{1}{T_0}\sum_{t=1}^{T_0} \ind{ M_t \in I} - \Pb[ M \in I] } < 8 \sqrt{\frac{\ln (1/\delta) }{T_0} }\right\}
                }\;.
            \]
            The family $\cI$ of all the subintervals of $[0,1]$ has VC dimension $2$ (see, e.g., \citet[Chapter~14.2]{MitzenmacherU17}). Therefore, $\cE_\delta^{T_0}$ is realized with probability at least $1-\delta$, via standard sample complexity bound for $\e$-samples (see, e.g., \citet[Theorem~14.15]{MitzenmacherU17}). This is summarized in the following lemma.
            \begin{lemma}
            \label{lem:VC}
                For every $T_0 \in \N$ and $\delta \in (0,1)$, we have 
                    $\Pb[\cE_\delta^{T_0}] \ge 1 - \delta$.
            \end{lemma}
            For the sake of readability, we introduce the following notation:
            \begin{notation}
                \label{not:k(b)}
                Let $\cX = \{ x_0, \dots, x_K \}$ be any grid with $0 = x_0 < x_1 < \dots < x_K \le 1$, we denote by $k_\cX \colon [0,1] \to \{0,1,\dots,K\}$ the function that maps each $b\in[0,1]$ to the unique $k$ such that $b\in[x_k, x_{k+1})$, with the convention that $x_{K+1} = 2$.
            \end{notation}

            We now prove a lemma that allows us to control the expected cumulative utility of any bid in $[0,1]$ with that of the best bid in a discretization (without relying on any smoothness assumption).

            \begin{lemma}
            \label{l:discret}
                Consider any finite grid $\cX = \{ x_0, \dots, x_K \}$, with $0 = x_0 < x_1 < \dots < x_K \le 1$, and assume that the process $M,M_1,M_2,\dots$ of the highest competing bids form an i.i.d.\ sequence. For all $b\in[0,1]$ and $T_0,T_1 \in \N$ with $T_0 < T_1$, $\bbE\lsb{ \sum_{t=T_0+1}^{T_1} \util_t(b) }$ is at most
                \[
                    \bbE \lsb{ \sum_{ t=T_0+1 }^{T_1} \util_t \brb{ x_{k_\cX(b)} } } + (T_1-T_0) \Pb \bsb{ x_{k_\cX(b)} < M < x_{k_\cX(b)+1} } \:.
                \]
            \end{lemma}
            \begin{proof}
                Fix any $b\in[0,1]$, $T_0,T_1 \in \N$ with $T_0 < T_1$, and a time step $t\in\{T_0+1, \dots, T_1\}$. 
                Then
                \begin{align*}
                &
                   \bbE\bsb{\util_t(b)} 
                = 
                    \bbE\bsb{(V_t-b)\I\{b \ge M_t \}}
                \\
                &
                \le 
                    \bbE\bsb{(V_t-x_{k_\cX(b)})\brb{ \I\{x_{k_\cX(b)} \ge M_t\} + \I\{b \ge M_t > x_{k_\cX(b)} \} } }
                \\
                &
                \le 
                    \bbE\bsb{\util_t(x_{k_\cX(b)})} + \Pb[x_{k_\cX(b)} < M_t \le b]
                \\
                &
                \le 
                    \bbE\bsb{\util_t(x_{k_\cX(b)})} + \Pb[x_{k_\cX(b)} < M_t < x_{k_\cX(b)+1}] \;.
            \end{align*}
            Summing over all times $t$ and recalling that $M_t$ and $M$ share the same distribution, yields the conclusion.
            \end{proof}

            As a corollary of \Cref{l:discret,lem:VC} we obtain a similar discretization error guarantee when the grid of points $\cX$ is random.

            \begin{lemma}
                \label{l:discret-rand}
                Fix any $T_0\in\N$ and $\delta \in (0,1)$.
                Let $\cX = \{ X_0, \dots, X_K \}$ be a random set containing a random number $K$ of points satisfying 
                $0 = X_0 < X_1 < \dots < X_K \le 1$.
                Assume that the random variables $K, X_0,X_1, \dots, X_{K+1}$ are $\cH_{T_0}$-measurable, where $\cH_{T_0}$ is the history up to and including time $T_0$. 
                Assume that the process $(V_1,M_1), (V_2,M_2),\dots$ of the valuations/highest competing bids form an i.i.d.\ sequence.
                Then, for all $b\in[0,1]$ and $T_1 \in \N$ with $T_1 > T_0$, we have:
                \begin{align*}
                    &\bbE\lsb{ \sum_{t=T_0+1}^{T_1} \util_t(b) }
                \le
                    \bbE \lsb{ \sum_{ t=T_0+1 }^{T_1} \util_t \brb{ X_{k_\cX(b)} } }
                \\
                 &+ (T_1-T_0) \lrb{ \! \tfrac{1}{T_0} \!\sum_{t=1}^{T_0} \Pb \lsb{ X_{k_\cX(b)} < M_t < X_{k_\cX(b)+1} } + 8\sqrt{\tfrac{\ln(\nicefrac 1\delta)}{T_0}} + \delta }.
                \end{align*}
            \end{lemma}

    \begin{algorithm}[t!]
    \caption*{\semitransparent{} (\st)}\label{alg:semi_iid}
    \begin{algorithmic}[1]
        \State \textbf{input:} Time horizon $T$, bandit algorithm $\tilde{\cA}$ for gains in $[-1,1]$
        \State $T_0 \gets \lceil T^{\nicefrac 23} \rceil $
        \State Run \collect with horizon $T_0$ and obtain $X_0, X_1, \dots, X_K$
        \State Initialize $\tilde{\cA}$ on $K+1$ actions (one for each candidate bid $X_i$) and $T-T_0$ as time horizon
        \For{each round $t = T_0 + 1, T_0 + 2, \dots, T$}
            \State Receive from $\tilde{\cA}$ the bid $B_t = X_{I_t}$ for some $I_t \in \{0,1,\dots, K\}$
            \State Post bid $B_t$ and observe feedback $Z_t$
            \State Reconstruct $\util_t(B_t)$ from $Z_t$ and feed it to $\tilde{\cA}$
        \EndFor
    \end{algorithmic}
    \end{algorithm}

We are now ready to present the main theorem of this section.
    
        \begin{theorem}
        \label{thm:upper-iid-general-semi-transparent}
            Consider the problem of repeated bidding in first-price auctions in a stochastic i.i.d.\ environment with semi-transparent feedback. 
            Then there exists a learning algorithm $\cA$ such that
            \[  
                R_T(\cA) \le 16 \brb{ 13 + \sqrt{\ln T} } T^{\nicefrac 23} \;.    
            \]
        \end{theorem}
        \begin{proof}
            We prove that $\semitransparent$ yields the desired bound when its learning routine $\tilde{\cA}$ is (a rescaled version of) MOSS \citep{audibert2009minimax}: since MOSS is designed to run with gains in $[0,1]$ while the utilities we observe are in $[-1,1]$, we first apply the reward transformation $x \mapsto \frac{x+1}{2}$ to the observed utilities. This costs a multiplicative factor of $2$ on the regret guarantees of MOSS. 
            Leveraging the fact that the empirical frequency between two consecutive $X_k$ and $X_{k+1}$ generated by \collect is at most $\nicefrac 2{\sqrt{T_0}}$ by design and applying \Cref{l:discret-rand} with $T_1 = T$ to the random variables $X_0,X_1,\dots,X_K$, we get, for all $b\in[0,1]$, that
            \[
                    \bbE\lsb{ \sum_{t=T_0+1}^{T} \util_t(b) } \le \bbE \lsb{ \sum_{ t=T_0+1 }^{T} \util_t \brb{ X_{k_\cX(b)} } }
                    + (T-T_0) \lrb{ \frac{2}{\sqrt{T_0}} + 8\sqrt{\frac{\ln(1/\delta)}{T_0}} + \delta }
                =
                    (\star) \;.
            \]
            Now, applying the tower rule to the expectation on the right-hand side conditioning to the history $\cH_{T_0}$ up to time $T_0$, we can use the fact that the regret of the rescaled version of MOSS is upper bounded by $98\sqrt{(K+1) (T-T_0)}$ and the number of points $K+1$ collected by \collect{} is at most $\sqrt{T_0} + 1$ to obtain 
            \begin{align*}
            &
                (\star)
            \le
                \bbE \lsb{ \sum_{ t=T_0+1 }^{T} \util_t (B_t) } 
                + 98\sqrt{( \sqrt{T_0}+1)(T-T_0) }
                + (T-T_0) \lrb{ \frac{2}{\sqrt{T_0}} + 8\sqrt{\frac{\ln(1/\delta)}{T_0}} + \delta } \;.
            \end{align*}
            Finally, tuning $\delta = 1/T_0$, upper bounding the cumulative regret over the first $T_0$ rounds with $T_0$, and recalling that $T_0 = \lceil T^{\nicefrac 23} \rceil$, yields the conclusion.
        \end{proof}

        \subsubsection[\texorpdfstring{A $T^{\nicefrac 23}$ lower bound for the smooth i.i.d.\ environment}{A T\^ 23 lower bound for the smooth i.i.d.\ environment}]{A $T^{\nicefrac 23}$ Lower Bound for the Smooth i.i.d.\ Environment}

            We prove that the $\tilde O(T^{\nicefrac 23})$ bound achieved by \semitransparent \xspace is indeed optimal, up to logarithmic terms. Our lower bound consists in carefully embedding into our model a hard multiarmed bandit instance with $K = \Theta(T^{\nicefrac 13})$ arms, which entails a lower bound of order $\Omega(\sqrt{KT}) = \Omega(T^{\nicefrac 23})$. This proof agenda involves various challenges: we want to embed a discrete construction of $K$ independent actions into our continuous framework, where the utilities of different bids are correlated, while enforcing smoothness. Furthermore, the semi-transparent feedback is richer than the bandit one.
            We report here a proof sketch and refer the interested reader to \Cref{app:lower-iid-smooth-semi-transparent} for the missing details.

            \begin{theorem}
            \label{thm:lower-iid-smooth-semi-transparent}
                Consider the problem of repeated bidding in first-price auctions in a stochastic i.i.d.\ $\sigma$-smooth environment with semi-transparent feedback, for $\sigma \in (0,\nicefrac 1{66}]$. Then, any learning algorithm $\cA$ satisfies, for $T \ge 8$,
                \[
                    R_T(\cA)
                \ge
                    \frac{3}{10^4}T^{\nicefrac 23} \;.
                \]
            \end{theorem}

         \begin{figure*}
                \centering
                \begin{tikzpicture}[scale=3]
                    \draw[gray, dashed] (0,1/4) node[black, left] {$\nicefrac14$} -- (7/8,1/4);
                    \draw[gray, dashed] (0,0.33) node[black, left] {$w$} -- (7/8,0.33);
                    \draw[gray, dashed] (0,3/4) node[black, left] {$\nicefrac34$} -- (7/8,3/4);
                    \draw[gray, dashed] (0,7/8) node[black, left] {$\nicefrac78$} -- (1,7/8);
                    \fill[yellow] (7/8,0) rectangle (1,1/4);
                    \draw (1.02, 0.36) -- (1.04, 0.36) -- (1.04, 0.3) -- (1.02, 0.3);
                    \draw (1.04, {0.33}) node[right] {$\Theta(\e)$};
                    \filldraw[yellow] (7/8,3/4) -- (1,7/8) -- (1,3/4) -- cycle;
                    \fill[green] (7/8,1/4) rectangle (1,0.3);
                    \fill[green] (7/8,0.36) rectangle (1,3/4);
                    \fill[blue] (7/8,0.3) rectangle (15/16,0.33);
                    \fill[blue] (15/16,0.33) rectangle (1,0.36);
                    \fill[red] (7/8,0.33) rectangle (15/16,0.36);
                    \fill[red] (15/16,0.3) rectangle (1,0.33);
                    \draw (7/8,1/4) -- (1,1/4)
                        (7/8,0.3) -- (1,0.3)
                        (7/8,0.33) -- (1,0.33)
                        (7/8,0.36) -- (1,0.36)
                        (15/16,0.36) -- (15/16,0.3) 
                        (7/8,3/4) -- (1,3/4)
                        (7/8,0) -- (7/8,3/4) --(1,7/8) -- (1,0)
                    ;
                    \draw[<->] (0,1.1) node[left]{$m$} -- (0,0) -- (1.1,0) node [right] {$v$};
                    \draw (7/8,0) node[below] {$\frac78$};
                    \draw (1,0) node[below] {$1$};

                    \fill[blue] (1.5, 0.2) rectangle (2, 0.44);
                    \fill[blue] (2, 0.44) rectangle (2.5, 0.68);
                    \fill[red] (1.5, 0.44) rectangle (2, 0.68);
                    \fill[red] (2, 0.2) rectangle (2.5, 0.44);
                    \draw (1.5, 0.44) -- (2.5, 0.44)
                        (2, 0.2) -- (2, 0.68);
                    \draw[thick, gray] (1, 0.3) -- (1.5, 0.2)
                        (1, 0.36) -- (1.5, 0.68);
                    \draw[thick, gray] (7/8, 0.3) rectangle (1,0.36);
                    \draw[thick, gray] (1.5, 0.2) rectangle (2.5,0.68);
                    \draw[white] (2.25, 0.31) node{$R^1_{w,\e}$};
                    \draw[white] (2.25, 0.55) node{$R^2_{w,\e}$};
                    \draw[white] (1.75, 0.31) node{$R^3_{w,\e}$};
                    \draw[white] (1.75, 0.55) node{$R^4_{w,\e}$};
                \end{tikzpicture}
                \qquad\qquad
                \begin{tikzpicture}[
                declare function={
                    func(\x)
                = 
                    (\x < 0) * (0)
                    + and(\x >= 0, \x < 1/4) * \x * ( 1/2 + ( 1 - 4 * \x ) * ln(6/5) )
                    + and(\x >= 1/4, \x < 3/4) * ( 1/8 )
                    + and(\x >= 3/4, \x < 7/8) * ( - 4 * \x * \x + 6 * \x - 17/8 )
                    + and(\x >= 7/8, \x < 1) * ( 15/16 - \x )
                    + (\x >= 1) * (0)
                   ;
                  },
                declare function={
                    tent(\w,\eps,\x)
                =
                    (\x < \w - \eps) * (0)
                    + and(\x >= \w - \eps, \x < \w) * ( 1 + ( \x - \w ) / \eps )
                    + and(\x >= \w, \x < \w + \eps) * ( 1 - ( \x - \w ) / \eps )
                    + (\x >= \w + \eps) * (0)
                    ;
                  },
                  scale = 3
                ]
                \draw[gray!25!white, very thin] (1/4, {3 * func(1/4)}) -- (1/4, 0);
                \draw[gray!25!white, very thin] (3/4, {3 * func(3/4)}) -- (3/4, 0);
                \draw[gray!25!white, very thin] (0.30, 0.6) -- (0.30, 0);
                \draw[gray!25!white, very thin] (0.33, { 3 * func( 0.33 ) + 30 * 0.33/144 *
                tent( 0.33, 0.03, 0.33 ) } ) -- (0.33, 0);
                \draw[gray!25!white, very thin] (0.36, 0.6) -- (0.36, 0);
                \draw[gray!25!white, very thin] (0.3, { 3 * func( 0.3 ) } ) -- (0, { 3 * func( 0.3 ) } );
                \draw[gray!25!white, very thin] (0.33, { 3 * func( 0.33 ) + 30 * 0.33/144 * tent( 0.33, 0.03, 0.33 ) } ) -- (0, { 3 * func( 0.33 ) + 30 * 0.33/144 * tent( 0.33, 0.03, 0.33 ) } );
                \draw (-0.02, { 3 * func( 0.3 ) }) -- (-0.04, { 3 * func( 0.3 ) } ) -- (-0.04, { 3 * func( 0.33 ) + 30 * 0.33/144 * tent( 0.33, 0.03, 0.33 ) } ) -- (-0.02, { 3 * func( 0.33 ) + 30 * 0.33/144 * tent( 0.33, 0.03, 0.33 ) } );
                \draw (-0.04, { ( 3 * func( 0.33 ) + 30 * 0.33/144 * tent( 0.33, 0.03, 0.33 ) + 3 * func( 0.3 ) ) / 2  } ) node[left] {$\Theta(\e)$}; 
                \draw (0.3, 0.62) -- (0.3, 0.64) -- (0.36, 0.64) -- (0.36, 0.62);
                \draw (0.33, 0.64) node[above] {$\Theta(\e)$};
                \draw[domain = 0:1, thick, densely dotted, samples = \nsamples] plot (\x, { 3 * func( \x ) });
                \draw[domain = 0:1, red, samples = \nsamples] plot (\x, { 3 * func( \x ) + 30 * 0.33/144 *
                tent( 0.33, 0.03, \x ) });
                \draw[<->] (0,1.1) node[left]{$m$} -- (0,0) -- (1.1,0) node [right] {$v$};
                \draw (1/4, 0) -- (1/4, -0.02)
                    (0.33, 0) -- (0.33, -0.02)
                    (3/4, 0) -- (3/4, -0.02)
                ;
                \draw (1/4, 0) node[below] {$\frac{1}{4}$}
                    (0.33, 0) node[below] {$w$}
                    (3/4, 0) node[below] {$\frac{3}{4}$}
                ;
                \end{tikzpicture}
                \caption{Left: The support of the base density $f$ lies inside the yellow and green regions. The perturbation $g_{w,\e}$ of $f$ occurs inside the green region, where the four rectangles $R^1_{w,\e},\dots, R^4_{w,\e}$ (in red and blue) lie. 
                Right: The corresponding qualitative plots of $b \mapsto \bbE[\util_t(b)]$ (black, dotted) and $p \mapsto \bbE^{w,\e}[\util_t(b)]$ (red, solid).%
                }
                \label{fig:lowerBoundSemiTransp}
            \end{figure*}
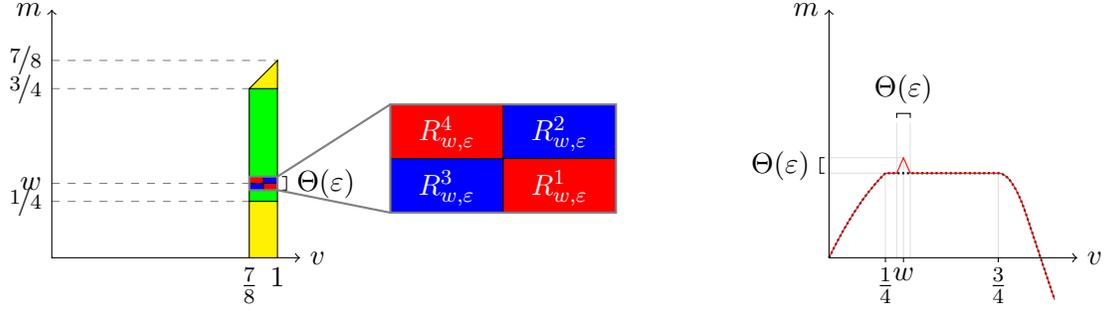

        \begin{proof}[Proof sketch]
            Define, for all $v,m\in[0,1]$, the density
            \[
                f(v,m)
            =
                \I_{\left[\tfrac 78, \, 1\right]}(v) \lrb{ \frac{1}{(v-m)^2} \I_{\left[\tfrac 14, \, v-\tfrac 18\right]}(m) + \frac{4}{v-\nicefrac 14} \I_{ \left[0 , \, \tfrac 14\right) } (m) }\footnote{Note, we use the notation $\I_A(x)$ to denote the indicator function that has value $1$ when $x \in A$, and $0$ otherwise.}.
            \]
            
            Let $\Pb^0$ be a probability measure such that $(V,M),(V_1,M_1)$, $ \dots$ is a $\Pb$-i.i.d.\ sequence where each pair $(V,M)$ has common probability density function $f$. 
            Denoting by $\bbE^0$ the expectation with respect to $\Pb^0$, we have, for any bid $b\in[0,1]$ and any time step $t$
            \begin{align*}
                \bbE^0\bsb{ \util_t(b) }
                    =
                    &b\lrb{\tfrac 12 + (1-4b) \ln \tfrac65 } \I_{\left[0,\tfrac14\right)}(b) 
                    +
                    \frac18 \I_{\left[\tfrac14, \tfrac34\right)} (b)
                    \\
                    &-
                    \lrb{4b^2 - 6b + \tfrac{17}8} \I_{\left[\tfrac34, \tfrac78\right) } (b) 
                    +
                    \lrb{ \tfrac{15}{16} - b } \I_{\left[\tfrac 78 , 1\right]} (b) \;.
            \end{align*}
               
            This function grows with $b$ on $[0,\nicefrac 14)$, has a plateau of maximizers $[\nicefrac 14, \nicefrac 34]$, then decreases on $(\nicefrac 34,1]$ (see \Cref{fig:lowerBoundSemiTransp}, right). We introduce the perturbation space $\Xi$: 
            \[
            \Xi = \bcb{ (w,\e) \in [0,1]^2 : w-\e \ge \tfrac 14 \text{ and } w+\e \le \tfrac 34 }
            \]
            and define, for all $(w,\e)\in \Xi$, the four rectangles 
            \begin{align*}
                R^1_{w,\e} &= [\nicefrac {15}{16}, \, 1] \times [w-\e, \, w), \quad R^2_{w,\e} = [\nicefrac {15}{16}, \, 1] \times [w, \, w + \e ),\\
                R^3_{w,\e} &= [\nicefrac {7}{8}, \, \nicefrac {15}{16}) \times [w -  \e, \, w), \,\,\, R^4_{w,\e} = [\nicefrac {7}{8}, \, \nicefrac {15}{16}) \times [w , \, w + \e ).
            \end{align*}
            For all $(w,\e)\in \Xi$, we introduce the probability density function $f_{w,\e}$ as follows
            $
                f_{w,\e}
            =
                f + g_{w,\e}
            $, where the perturbation $g_{w,\e}$ is defined as follows 
            \[
                g_{w,\e} (v,m)
            =
                \frac{16}{9} \brb{ \I_{R^1_{w,\e} \cup R^4_{w,\e}}(v,m) - \I_{R^2_{w,\e} \cup R^3_{w,\e}}(v,m)
                }
              \;.
            \]
            We refer to the left plot in \Cref{fig:lowerBoundSemiTransp} for a visualization of the support of the $f_{w,\e}$. 
            For all $(w,\e)\in\Xi$, let $\Pb^{w,\e}$ be a probability measure such that $(V,M),(V_1,M_1), (V_2,M_2), \dots$ is a $\Pb^{w,\e}$-i.i.d.\ sequence where each pair $(V,M)$ has common probability density function $f_{w,\e}$. 
            Denoting by $\bbE^{w,\e}$ the expectation with respect to $\Pb^{w,\e}$, we have, for any bid $b\in[0,1]$ and any $t$
            \[
                \bbE^{w,\e}\bsb{ \util_t(b) }
            =
                \bbE^0 \bsb{ \util_t(b) } + \frac{\e}{144} \Lambda_{w,\e}(b)
            \]
            where $\Lambda_{u,r}$ is the tent map centered at $u$ with radius $r$ defined as $\Lambda_{u,r}(x) = \max\lcb{1-\fracc{|x-u|}{r}, 0}$. 
            In words, in a perturbed scenario $\Pb^{w,\e}$ the expected utility is maximized at the peak of a spike centered at $w$ with length and height $\Theta(\varepsilon)$ perturbing the plateau area $[\nicefrac 14, \, \nicefrac 34]$ of maximum height (see \Cref{fig:lowerBoundSemiTransp}, right). 
            Define, for all times $t\in\N$, the feedback function $\psi_t \colon [0,1] \to \brb{ [0,1] \times \{\star\} } \cup \brb{ \{\star\} \times [0,1] }$, as follows:
            \[
             b \mapsto
                \begin{cases}
                    (V_t,\star) & \text{ if } b \ge M_t \\
                    (\star, M_t) & \text{ if } b < M_t \\
                \end{cases}
            \]
            and note that, in our semi-transparent feedback model, the feedback $Z_t$ received after bidding $B_t$ at time $t$ is $\psi_t(B_t)$.
            Crucially, for each $(w,\e) \in \Xi$ and each $b\in[0,1]\m[w-\e,w+\e]$, the distribution of $\psi_t(b)$ under $\Pb^{w,\e}$ coincides with the distribution of $\psi_t(b)$ under $\Pb^0$. In push-forward notation(for a refresher on push-forward measures, see \Cref{s:appe-measure}), it holds that 
            \begin{equation}
            \label{eq:same-stuff}
                \Pb_{\psi_t(b)}^{w,\e} = \Pb^0_{\psi_t(b)} \;.
            \end{equation}
            Now, let $K\in \N$, $\e = \nicefrac 1{(4K)}$, $w_k = \nicefrac 14 + (2k-1) \e$ and $\Pb^k = \Pb^{w_k,\e}$ (for each $k\in[K]$).
            At a high level, we built a problem with two crucial properties: (i) we know in advance the region where the optimal bid belongs to (i.e., the interval $[\nicefrac 14, \nicefrac 34]$), but (ii) when the underlying scenario is determined by the probability measure $\Pb^k$, the learner has to detect inside this potentially optimal region where a spike of height (and length) $\Theta(\e)$ occurs (to avoid suffering suffer $\Omega( \e T )$ regret). This last task can be accomplished only by locating where the perturbation in the base probability measure occurs, which, given the feedback structure, can only be done by playing in the interval $[w_k-\e,w_k+\e)$ if the underlying probability is $\Pb^k$, suffering instantaneous regret of order $\e$ whenever the underlying probability is $\Pb^j$, with $j \neq k$.
            Given that we partitioned the potentially optimal region $[\nicefrac 14, \nicefrac 34]$ into $\Theta(\nicefrac{1}{\e})$ disjoint intervals where these perturbations can occur, the feedback structure implies that each of these intervals deserves its dedicated exploration. 

            To better highlight this underlying structure, in \Cref{app:lower-iid-smooth-semi-transparent}, we show that our problem is not easier than a simplified $K$-armed stochastic bandit problem, where the instances we consider are determined by the probability measures $\Pb^1, \dots, \Pb^K$. 
            In this bandit problem, when the underlying probability measure is induced by some $\Pb^k$, the corresponding arm $k$ has an expected reward $\Theta(\e)$ larger than the others.
            Then, via an information-theoretic argument, we can show that any learner would need to spend at least order of $\nicefrac 1{\e^2}$ rounds to explore each of the $K$ arms (paying $\Omega(\e)$ each time) or else, it would pay a regret $\Omega(\e T)$. 
            Hence, the regret of any learner, in the worst case, is lower bounded by $\Omega\brb{ \frac{K}{\e^2} \e + \e T } = \Omega \brb{ K^2 + \tfrac TK } $ (recalling our choice of $\e = \nicefrac 1{(4K)}$). 
            Picking $K = \Theta(T^{\nicefrac 13})$ yields a lower bound of order $T^{\nicefrac 23}$. For all missing technical details, see \Cref{app:lower-iid-smooth-semi-transparent}.
        \end{proof}

    \subsection{I.I.D.\ -- Transparent/Full Feedback}
    \label{sec:iid-semi-full}

        This section completes the study of the stochastic i.i.d.\ environment by determining the minimax regret when the learner has access to full or transparent feedback.

        \subsubsection[\texorpdfstring{A $\sqrt{T}$ upper bound for the i.i.d.\ environment}{A sqrt T upper bound for the i.i.d.\ environment}]{A $\sqrt{T}$ Upper Bound for the i.i.d.\ Environment}
        While with semi-transparent feedback, the learning algorithm has to rely on dummy bids $B_1=\dots=B_{T_0} = 0$ to gather information about the distribution of the highest competing bids, with the transparent one, this information is collected for free at each bidding round.
        To use this extra information, we present a wrapper \expthreefpawrapper{} (for a sequence of base learning algorithms for the transparent feedback model) whose purpose is restarting the learning process with a geometric step to update the set of candidate bids. 
        We assume that each of the wrapped base algorithms $\tilde \cA_\tau$ can take as input any finite subset $\cX \subset [0,1]$ and returns bids in $\cX$. 
        Furthermore, for all $T'$, we let $\cR_{T'}(\tilde \cA_\tau,\cX)$ be an upper bound on the regret over $T'$ rounds of $\tilde\cA_\tau$ with input $\cX$ against the best fixed $x\in \cX$. 
        Formally, we require that for any two times $T_0 < T_1$ such that $T' = T_1 - T_0$, the quantity $\cR_{T'}(\tilde \cA_\tau,\cX)$ is an upper upper bound on
        $
            \max_{x\in\cX} \bbE\bsb{ \sum_{t=T_0+1}^{T_1} \util_t(x) - \sum_{t=T_0+1}^{T_1} \util_t(B_t) }
        $, where $B_t \in \cX$ is the sequence of prices played by $\tilde \cA_\tau$ (with input $\cX$) when started at round $t=T_0+1$ and ran up to time $T_1$.
        Without loss of generality, we assume that $T' \mapsto \cR_{T'}(\tilde \cA_\tau,\cX)$ is non-decreasing.

        \begin{algorithm}[ht]
        \caption*{\expthreefpawrapper{} (\expthreefpawrapperLong{}) }\label{alg:expthreefpaWrapper}
        \begin{algorithmic}[1]
            \State \textbf{input:} Base algorithms $\tilde \cA_1 , \tilde\cA_2, \dots$  \State \textbf{initialization:} $s\gets 0$
            \For{each epoch $\tau = 1, 2,\dots$}
                \State $\cX_\tau \gets \{0\} \cup \{ M_1, \dots, M_s\}$ (with $\cX_1 = \{0\}$) 
                \State Start $\tilde\cA_\tau$ with input $\cX_{\tau}$ and run it for $t = s+1, \dots, s+ 2^{\tau -1}$
                \State Update $s \gets s+ 2^{\tau -1}$
            \EndFor
        \end{algorithmic}
        \end{algorithm}

        \begin{proposition}
        \label{prop:upper-iid-general-transparent}
            Consider the problem of repeated bidding in first-price auctions in a stochastic i.i.d.\ environment with transparent feedback. 
            Then the regret of \expthreefpawrapper{} run with base algorithms $\tilde\cA_1, \tilde\cA_2,\dots$ satisfies
            \[  
                R_T(\mathrm{\expthreefpawrapper{}}) \le \!\!\! \sum_{\tau=2}^{\lceil \log_2(T+1) \rceil} \!\!\!\cR_{2^{\tau - 1}} \brb{ \tilde \cA_\tau,\cX_\tau } + 3 + 16 \brb{ \sqrt{2}+2 } \sqrt{ T \ln T} \;.    
            \]
        \end{proposition}
        \begin{proof}
            Fix an arbitrary epoch $\tau \in \bcb{2,\dots, \lce{ \log_2 (T+1) } }$; we want to bound the regret suffered there by \expthreefpawrapper{} using \Cref{l:discret-rand}. Using the notation of the lemma, let $\cX = \cX_\tau$,
            $K + 1 = \labs{\cX}$,
            $T_0 = \sum_{\tau'=1}^{\tau-1}2^{\tau'-1}=2^{\tau-1}-1$ (the time passed from the beginning of epoch $1$ up to and including the end of epoch $\tau-1$),
            $T_1 = \min\{T_0+2^{\tau-1},T\}$ (the end of epoch $\tau$),
            and let $X_0 < X_1 < \dots < X_K$ be the distinct elements of $\cX$ in increasing order, where we note that $X_0 = 0$, $X_K\le 1$, and we set $X_{K+1}=2$.
            Let $\cH_{T_0}$ be the history, including time $T_0$. 
            
            Applying \Cref{l:discret-rand} (together with the fact that the empirical frequency between any two consecutive values $X_k$ and $X_{k+1}$ is $0$ by design), and exploiting the monotonicity of $T' \mapsto \cR_{T'}(\tilde \cA_\tau,\cX_\tau)$ for the last epoch (if $T_0+2^{\tau-1} > T$), we obtain, for all $b\in[0,1]$ and $\delta\in(0,1)$,
            \begin{align*}
                \bbE&\lsb{ \sum_{t=T_0+1}^{T_1} \!\!\!\util_t(b) }
                \le \!\!\!
                     \sum_{ t=T_0+1 }^{T_1} \!\!\!\bbE \lsb{\util_t \Brb{ X_{k_\cX(b)} } }
                    + 2^{\tau-1} \lrb{ 8\sqrt{\tfrac{\ln(\nicefrac 1\delta)}{T_0}} + \delta }
                 \\
                 &
                 \;
                 \le
                     \sum_{ t=T_0+1 }^{T_1} \!\!\!\bbE \lsb{ \util_t (B_t) }
                    + \cR_{2^{\tau - 1}} \brb{ \tilde \cA_\tau,\cX_\tau }
                    + 2^{\tau-1} \lrb{ 8 \sqrt{\tfrac{\ln(\nicefrac 1\delta)}{2^{\tau-1}-1}} + \delta }.
            \end{align*}
            Summing over epochs $\tau \in \bcb{2,\dots, \lce{ \log_2 (T+1) } }$, upper bounding by $1$ the regret incurred in the first epoch, and tuning $\delta = \nicefrac 1T$ yields the conclusion.
        \end{proof}

        Now we are only left to design appropriate base algorithms $\tilde\cA_1,\tilde\cA_2,\dots$ for the transparent feedback to wrap \expthreefpawrapper{} around.
        

        \paragraph{The \expthreefpa{} algorithm.}
        To this end, w%
        e introduce the \expthreefpa{} algorithm (designed to run with transparent feedback), which
        borrows ideas from online learning with feedback graphs \citep{AlonCGMMS17}. 
        Similar algorithms for related settings have been previously proposed by \citet{WeedPR16} and \citet{FengPS18}. 
        For the familiar reader, note that our setting can be seen as an instance of online learning with strongly observable feedback graphs. 
        In contrast to a black-box application of feedback-graph results, we shave off a logarithmic term (in the time horizon) by using a dedicated analysis. 
        For any $x \in [0,1]$, we denote by $\delta_x$ the Dirac distribution centered at $x$.
        
        \begin{algorithm}[ht]
        \caption*{\expthreefpa{}}\label{alg:expthreefpa}
        \begin{algorithmic}[1]
            \State \textbf{input:} Finite $\cX\s[0,1]$ with maximum $\xbar$, exploration rate $\gamma\in(0,1)$
            \State For all $x\in\cX$, let $w_1(x) \gets 1$
            \For{each round $t = 1, 2, \dots$}
                \State Post bid $B_t \sim p_t \gets  (1-\gamma) \frac{w_t}{\lno{w_t}_1} + \gamma\delta_{\xbar}$
                \State For all $x\in\cX$, define the reward estimate:
                \[
                \ghat_t(x) \gets (V_t - x) \I \{ x \ge M_t \} \frac{ \I \{ M_t \le B_t \} } { \sum_{y\ge M_t} p_t(y) }
                \]
                \label{state:transp-fb-line}
                \State For all $x\in\cX$, update the weight:
                \[
                w_{t+1}(x) \gets w_t(x) \exp\brb{ \gamma \ghat_t(x) } 
                \]
            \EndFor
        \end{algorithmic}
        \end{algorithm}

        Note that the transparent feedback is sufficient to compute the reward estimates in \Cref{state:transp-fb-line}.

        \begin{proposition}
            \label{prop:routine}
            Let $ \cX \subset [0,1]$ be a finite set, $T\in \N$ a time horizon, and tune the exploration rate as $\gamma = \sqrt{\fracc { \ln ( \labs{\cX} ) }{ (e-1)T }}$. 
            Then, the regret of \expthreefpa{} against the best fixed bid in $\cX$ is
            \[
                \max_{x\in\cX} \bbE\lsb{ \sum_{t=1}^T \util_t(x) - \sum_{t=1}^T \util_t(B_t) }
            \le
                2 \sqrt{(e-1) \ln \brb{ \labs{\cX} } T }
            \]
        \end{proposition}
        \begin{proof}
            Let $\gamma >0$.
            Notice that, for each $t \in \N$, it holds that $\sum_{y \ge M_t} p_t(y) \ge \gamma$. It follows, for each $x \in \cX$ and $t \in \N$, that $\gamma \ghat_t(x) \le 1$, and hence
            \[
                \exp(\gamma \ghat_t(x)) \le 1 + \gamma \ghat_t(x) + (e-2) \gamma^2 \brb{\ghat_t(x)}^2\;.
            \]
            Then, for each $t \in \N$,
            \[
                \frac{\lno{w_{t+1}}_1}{\lno{w_{t}}_1}
            =
                \sum_{x \in \cX} \frac{w_t(x)}{\lno{w_{t}}_1}  \exp\brb{\gamma \ghat_t(x)}
            \le
                1+\sum_{x \in \cX} \frac{w_t(x)}{\lno{w_{t}}_1}  \Brb{\gamma \ghat_t(x) + (e-2) \gamma^2 \brb{\ghat_t(x)}^2}\;,
            \]
            which implies
            \[
                \ln\lrb{\frac{\lno{w_{t+1}}_1}{\lno{w_{t}}_1}}
            \le
                \sum_{x \in \cX} \frac{w_t(x)}{\lno{w_{t}}_1}  \Brb{\gamma \ghat_t(x) + (e-2) \gamma^2 \brb{\ghat_t(x)}^2}
            \le
                \frac{\gamma}{1-\gamma}\sum_{x \in \cX} p_t(x)  \Brb{ \ghat_t(x) + (e-2) \gamma \brb{\ghat_t(x)}^2}.
            \]
            Now, for each $t \in \N$, let $\cF_t$ be the $\sigma$-algebra generated by $p_t, V_t$ and $M_t$ and denote by $\bbE_t := \bbE[\cdot \mid \cF_t]$. First, notice that, for each $t \in \N$ and each $x \in \cX$
            \begin{align*}
                \bbE_t[\ghat_t(x)] = \util_t(x)\;, \quad 
                \bbE_t \lsb{ \sum_{x \in \cX} p_t(x)\ghat_t(x) } = \bbE[\util_t(B_t) \mid V_t, M_t]\;,
            \end{align*}
            and that
            \[
                \bbE_t\lsb{\sum_{x \in \cX} p_t(x)\brb{\ghat_t(x)}^2}
            \le
                \bbE_t\lsb{\sum_{x \in \cX} p_t(x) \frac{\I\{x\ge M_t\}\I\{M_t \le B_t\}}{ \brb{ \sum_{y \ge M_t} p_t(y) }^2}}
            =
                \bbE_t\lsb{\sum_{x \in \cX} p_t(x) \frac{\I\{x\ge M_t\}}{\sum_{y \ge M_t} p_t(y)}} = 1\;.
            \]
            It follows that, for each $x \in \cX$,
            \begin{align*}
                \bbE \lsb{ \sum_{t=1}^T \util_t(x)} - \ln \brb{ \labs{ \cX } }
            &=
                \bbE \lsb{ \sum_{t=1}^T \ghat_t(x)} - \ln \brb{ \labs{ \cX } }
            =
                \bbE\Bsb{\ln\brb{w_{T+1}(x)}} - \ln \brb{ \labs{ \cX } }
            \\
            &\le
                \bbE\lsb{\ln\lrb{\frac{\lno{w_{T+1}}_1}{\lno{w_{1}}_1}}}
            =
                \sum_{t=1}^T\bbE\lsb{ \bbE_t\lsb{ \ln\lrb{\frac{\lno{w_{t+1}}_1}{\lno{w_{t}}_1}}} }
            \\
            &
            \le
                \frac{\gamma}{1-\gamma} \lrb{\bbE \lsb{ \sum_{t=1}^T \util_t(B_t) } + (e-2)\gamma T} \;,
            \end{align*}
        which, after rearranging and upper bounding, yields
        \[
            \bbE \lsb{ \sum_{t=1}^T \util_t(x) - \sum_{t=1}^T \util_t(B_t) }
        \le
            \frac{\ln\lrb{|\cX|}}{\gamma} + (e-1)\gamma T\;.
        \]
        Selecting $\gamma$ as in the statement of the theorem leads to the conclusion.
        \end{proof}

        Putting together \Cref{prop:upper-iid-general-transparent,prop:routine} yields the desired rate.

        \begin{theorem}
            \label{thm:upper-iid-general-transparent}
            Consider the problem of repeated bidding in first-price auctions in a stochastic i.i.d.\ environment with transparent feedback. 
            Then there exists a learning algorithm $\cA$ such that \[
                R_T(\cA) \le 
                3 + 2\brb{ \sqrt{2}+2 }\brb{ \sqrt{2(e-1)} + 8 } \sqrt{T \ln T} \;.
            \]
            
        \end{theorem}
        \begin{proof}
            The statement of the theorem holds for  \expthreefpawrapper{} run with the base algorithm of each epoch $\tau$ being $\expthreefpa{}$ tuned with $\gamma = \gamma(\tau) = \sqrt{\fracc { \ln ( |\cX_\tau| ) }{ \lrb{ (e-1) 2^{\tau-1} } }}$. Substituting the guarantees of \Cref{prop:routine} into those of \Cref{prop:upper-iid-general-transparent} and recalling that $\labs{\cX_\tau} \le 2^{\tau-1}$ for each epoch $\tau = 2,3,\dots$, yields the desired bound.
        \end{proof}

        \subsubsection[\texorpdfstring{A $\sqrt{T}$ lower bound for the i.i.d.\ environment}{A sqrt T lower bound for the i.i.d.\ environment}]{A $\sqrt{T}$ Lower Bound for the i.i.d.\ Environment}
            
            We complement the positive result of \Cref{thm:upper-iid-general-transparent} with a matching lower bound of order $\sqrt{T}$. The idea underlying our hard instance is to embed the well-known lower bound for prediction with (two) experts into our framework: we construct two smooth distributions that are ``similar'' but have two different optimal bids whose performance is separated so that no learner can identify the correct distribution without suffering less than $\sqrt{T}$ regret.

\begin{figure}[ht!]
         \centering
         \includegraphics[width = 0.7\textwidth]{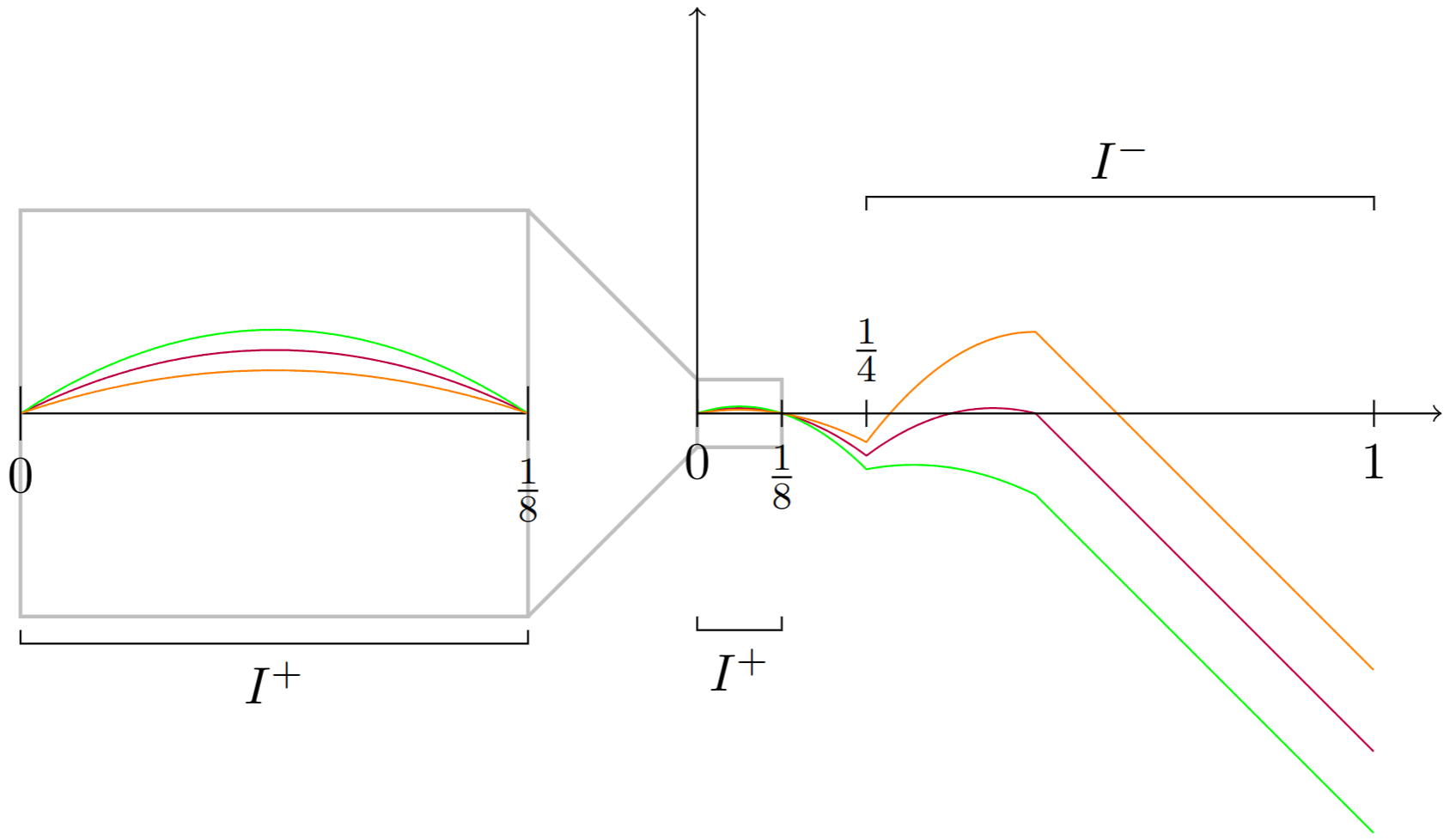}
         \caption{The expected utility function for three different distributions: $\Pb^0$ in purple, $\Pb^+$ in orange, and $\Pb^+$ in green. }
    \label{fig:supporting_full}
    \end{figure}

        \begin{theorem}
        \label{thm:lower-iid-smooth-full}
            Consider the problem of repeated bidding in first-price auctions in a stochastic i.i.d.\ $\sigma$-smooth environment with full feedback, for $\sigma \in (0,\nicefrac 19]$. Then, any learning algorithm $\cA$ satisfies
        \[
            R_T(\cA)
        \ge
            \frac1{2048} \sqrt{T} \;.
        \]
        \end{theorem}
        \begin{proof}
            We prove the theorem by Yao's principle: we show that there exists a distribution over stochastic $\sigma$-smooth environments such that any deterministic learning algorithm $\cA$ suffers $\Omega(\sqrt{T})$ regret against it, in expectation. We do that in two steps. First, for every $\e\in (0,\nicefrac 12)$ we construct a pair of $\nicefrac 19$-smooth distributions that are hard to discriminate for the learner. Then, we prove that, for the right choice of $\e$, any learner suffers the desired regret against a uniform mixture of them. For visualization, we refer to \Cref{fig:supporting_full}.

            As a tool for our construction, we introduce a baseline probability measure $\Pb^0$, such that the sequence $(V,M),(V_1,M_1),(V_2,M_2),\dots$ is $\Pb^0$-i.i.d., and $(V,M)$ has distribution $\Pb^0_{(V,M)}$ (for a refresher on push-forward measures, see \Cref{s:appe-measure}) whose pdf is
            \[
                f^0(v,m) = 8 \,(\I_{Q_+}(v,m) + \I_{Q_-}(v,m)),
            \]
            where $Q_+ = (0,\nicefrac{1}{4}) \times (0,\nicefrac{1}{4})$ and $Q_- = (\nicefrac{3}{4},1) \times (\nicefrac 14,\nicefrac{1}{2})$. 
            A convenient way to visualize this distribution is to draw a uniform random variable $U_t$ in the square $Q_+$ and then toss an unbiased coin. If the coin yields heads, then $(V_t,M_t)$ is equal to $U_t$, otherwise $(V_t,M_t)$ coincides with $U_t$ translated by $(\nicefrac 34, \nicefrac 14)$. With some simple computation, it is possible to explicitly compute the expected utility of posting any bid $b\in [0,1]$ when $(V_t,M_t)$ is drawn following the distribution $\Pb^0$ (and expectation $\mathbb{E}^0$):
            \[
            \mathbb{E}^0[\util_t(b)] = 
            \begin{cases}
                \frac b4 (1 - 8b) \quad & \text{ if $b \in [0,\nicefrac 14)$}\\
                -\frac 1{8} (16b^2 - 14b + 3) \quad & \text{ if $b \in [\nicefrac 14,\nicefrac 12)$}\\
                \tfrac 12(1 - 2b) \quad & \text{ if $b \in [\nicefrac 12,1]$}
            \end{cases}
            \]
            The function $\mathbb{E}^0[\util_t(b)]$ has two global maxima in $[0,1]$, of value $\nicefrac{1}{128}$, attained in $\nicefrac{1}{16}$ and $\nicefrac{7}{16}$ (see purple line in \Cref{fig:supporting_full}).

            For any $\e \in (0,\nicefrac 12)$, we also define two additional (perturbed) probability measures $\Pb^{\pm \e}$, such that the sequence $(V,M)$, $(V_1,M_1)$, $\dots$ is $\Pb^{\pm \e}$-i.i.d.\ and the distribution $\Pb^{\pm \e}_{(V,M)}$ of $(V,M)$ has density:
            \[
                f^{\pm \e}(v,m) = 8(1 \pm \e)\I_{Q_+}(v,m) + 8(1 \mp \e)\I_{Q_-}(v,m).
            \]
            Note, $||f^{\pm \e}||_{\infty} < 9$, while $||f^{0}||_{\infty} = 8$, therefore all the distributions considered in this proof are $\nicefrac 19$-smooth.
            To visualize these new perturbed distributions, recall the construction of $\Pb^0_{(V,M)}$ using the coin toss and the uniform random variable $U$: in this case, the coin is biased, and the probability of tails is $\nicefrac{(1\pm\e)}2$. It is possible to explicitly compute the expected utility under these perturbed distributions for any bid $b \in [0,1]$: $\mathbb{E}^{\pm \e}[\util_t(b)]$ is equal to
            \begin{equation}
            \label{eq:Epm}
            \begin{cases}
                \frac b4 (1 - 8b) \pm \e\frac{b}{4}(1-8b) \quad & \text{ if $b \in [0,\tfrac 14)$}\\
                -\frac 1{8} (16b^2 - 14b + 3) \pm \frac{\e}{4}(8b^2-11b+2) \quad & \text{ if $b \in [\tfrac 14,\tfrac 12)$}\\
                \tfrac 12(1 - 2b\mp \frac34 \e) \quad & \text{ if $b \in [\tfrac 12,1]$}
            \end{cases}     
            \end{equation}
            For visualization, we refer to \Cref{fig:supporting_full} (bottom). The crucial property of the distributions we constructed is that the instantaneous regret of not playing in the ``correct'' region is $\Omega(\e)$; formally we have the following result. For the sake of readability, we postpone the proof of this claim to \Cref{app:lower-iid-smooth-full}. 
            \begin{restatable}{claim}{suboptimality}\label{cl:suboptimality_gap}
                There exists two  disjoint intervals $I_+$ and $I_-$ in $[0,1]$ such that, for any $\e \in (0,\nicefrac 12)$ and any time $t$, the following hold:
                \[
                    \max_{x \in [0,1]}\mathbb E ^{\pm \e} [\util_t(x)] \ge \mathbb E ^{\pm \e} [\util_t(b)] +  \tfrac 1{128}\e, \text{ for all $b\notin I_{\pm}$}    
                \]
            \end{restatable}
            
            Since the two distributions are ``$\e$-close\footnote{In \Cref{app:lower-iid-smooth-full} we formally prove that their total variation is at most $\Theta(\e)$.}'', any learner needs at least $\nicefrac 1{\eps^2}$ rounds to discriminate which ones of the two distributions it is actually facing, paying each error with an instantaneous regret of $\Omega(\e)$ (\Cref{cl:suboptimality_gap}). All in all, any learner suffers a regret that is $\Omega(\e \cdot \tfrac{1}{\e^2} + \e T)$, which is of the desired $\Omega(\sqrt{T})$ order for the right choice of $\e \approx T^{-\nicefrac 12}$.
            
            As the last step of the proof, we formalize the above argument. Fix $\e = \nicefrac 1{(4\sqrt{T})}$ and rename $\Pb^{+\e}=\Pb^{1}$ and $\Pb^{-\e}=\Pb^{2}$. Similarly, denote with $I_1$ and $I_2$ the two intervals $I_+$ and $I_-$ as in the statement of \Cref{cl:suboptimality_gap}. 
            For each $j \in \{0,1,2\}$, consider the run of $\mathcal A$ against the stochastic environment which draws $(V_1,M_1),(V_2,M_2), \dots $ i.i.d.\ from $\mathbb{P}^j$. Let $N_1$ be the random variable that counts the number of times that algorithm $\cA$ posts a bid in $I_1$. Similarly, $N_2$ counts the number of times that it posts a bid in $I_2.$ For $i=1,2$, we have the following crucial relation between the expected value of $N_i$ under $\Pb^i$. Note, the results hold because the two distributions are so similar that the deterministic algorithm $\cA$ bids in the wrong region a constant fraction of the time steps. For the formal proof of we refer the reader to \Cref{app:lower-iid-smooth-full}.
            
            \begin{restatable}{claim}{lastStep}\label{cl:last_step}
                The following inequality holds: $       
                    \frac 12 \sum_{i=1,2}\Ei{N_i} \le \frac 34 T. 
                $
            \end{restatable}
            
            We finally have all the ingredients to conclude the proof. Consider an environment that selects uniformly at random either $\Pb^1$ or $\Pb^2$ and then draws the $(V_t,M_t)$ i.i.d.\ following it. We prove that the algorithm $\cA$ suffers linear regret against this randomized environment and, by a simple averaging argument, against at least one of them. Specifically, if $b^\star_i$ is the optimal bid in the scenario determined by $\Pb^i$, for $i \in \{1,2\}$, we have
            \begin{align*}
                R_T(\cA) &\ge \frac 12 \sum_{i=1,2}\mathbb{E}^{i}\left[{\sum_{t = 1}^{T} \util_t(b^\star_i) - \sum_{t = 1}^{T} \util_t(B_t)}\right]\\
                &\overset{(*)}{\ge} \frac {1}{1024\sqrt{T}} \sum_{i=1,2}\mathbb{E}^{i}\left[T - N_i\right]
                \overset{(\circ)}{\ge} \frac {1}{512\sqrt{T}} \left(T - \frac 34 T\right)
                = \frac {\sqrt{T}}{2048}
            \end{align*}
        where $(*)$ follows by \Cref{cl:suboptimality_gap} and choice of $\e$, and $(\circ)$ by \Cref{cl:last_step}.
        \end{proof}

\section{The Adversarial Setting}

    In this section we complete the perspective on repeated bidding in first-price auction by investigating the adversarial environment. In particular, we consider two models: the standard one, where the sequence $(V_1,M_1), (V_2,M_2), \dots$ is chosen upfront in a deterministic oblivious way, and the smooth environment, where the sequence $(V_1,M_1), (V_2,M_2), \dots$ is some $\sigma$-smooth stochastic process. In \Cref{sec:adv_bandit_smooth} we construct an algorithm achieving $T^{\nicefrac 23}$ regret in the bandit feedback model under the smoothness assumption; this result, together with the lower bound of the same order for the semi-transparent feedback (\Cref{thm:lower-iid-smooth-semi-transparent}) settles the problem for these two feedback regimes. Then, in \Cref{sec:adv_transparent_smooth} we provide another upper bound, namely an algorithm achieving $\sqrt T$ regret in the transparent feedback model under the smoothness assumption; this result, together with the lower bound of the same order for the semi-transparent feedback (\Cref{thm:lower-iid-smooth-full}) settles the problem for these two feedback regimes. Finally, in \Cref{sec:adv_full} we provide a lower bound proving that the non-smooth adversarial environment is too hard to learn, even when the learner has access to full feedback.

    \subsection{Smooth -- Bandit Feedback}\label{sec:adv_bandit_smooth}

        The smoothness assumption regularizes the objective function: if $(V_t,M_t)$ is smooth, then the expected utility is Lipschitz.  

        \begin{lemma}[Lipschitzness]
    \label{lem:lip}
        Let $(V_t,M_t)$ be a $\sigma$-smooth random variable in $[0,1]^2$. Then the induced expected utility function $\E{\util_t(\cdot)}$ is $\nicefrac 2\sigma$-Lipschitz in $[0,1]$:
    \begin{equation}
            \babs{\E{ \util_t(y) - \util_t(x)}} \le \frac{2}{\sigma} |y-x|, \quad \forall x,y \in [0,1].
        \end{equation}
    \end{lemma}
    \begin{proof}
        Let $x>y$ be any two bids in $[0,1]$, we have:
        \begin{align*}
            \babs{\mathbb E[ \util_t(x) - \util_t(y)]} &= \babs{\E{(V_t-x)\ind{M_t\le x} - (V_t-y)\ind{M_t\le y}}}  \\
            &= \babs{\E{(V_t-x) \ind{y < M_t \le x} + (x-y) \ind{M_t \le y}) }}\\
            &\le \Pb\bsb{M_t \in [x,y]} + (x-y) \le \tfrac{2}\sigma(x-y). \qedhere
        \end{align*} 
        \end{proof}

        Interestingly, we only need the marginal distribution of $M_t$ to be $\sigma$-smooth for the previous lemma to hold. This Lipschitzness property has the immediate corollary that any fine enough discretization of $[0,1]$ contains a bid whose utility is close the the optimal one. 
    
        \begin{lemma}[Discretization Lemma]
        \label{lem:discr_smooth}
            Let $\cX$ be any finite grid of bids in $[0,1]$, and let $\delta(\cX)$ be the largest distance of a point in $[0,1]$ to $\cX$ (i.e., $\delta(\cX) = \max_{p \in [0,1]}\min_{x \in \cX} |p-x|$), then if each pair of random variables $(V_1,M_1), \ldots, (V_T,M_T)$ is $\sigma$-smooth, we have the following:
            \[
                \sup_{b \in [0,1]}\bbE\lsb{\sum_{t=1}^T \util_t(b)} - \max_{x \in \cX}\bbE\lsb{\sum_{t=1}^T \util_t(x)}\le 3 \frac{\delta(\cX)}{\sigma} T \;.
            \]
        \end{lemma}
        \begin{proof}
            Fix any such sequence and let $b^\star$ a fixed bid such that 
            \begin{equation}
                \label{eq:sup}
                \sup_{b \in [0,1]}\bbE\lsb{\sum_{t=1}^T \util_t(b)} \le \bbE\lsb{\sum_{t=1}^T \util_t(b^\star)} + \frac{\delta(\cX)}{\sigma} T.
            \end{equation}
            If $b^\star$ is in $\cX$ there is nothing to prove, otherwise these exists $x^\star\in \cX$ such that $|b^\star - x^\star| \le \delta(\cX)$ (by definition of $\delta(\cX)$). It holds that
            \begin{align*}
            \sum_{t=1}^T \bbE\lsb{ \util_t(b^\star)- \util_t(x^\star)}
            \overset{(\mathrm{L})}{\le} \sum_{t=1}^T \frac{2}{\sigma}|b^\star - x^\star| 
                \le 2 \frac{\delta(\cX)}{\sigma} T.
            \end{align*}
            where $(\mathrm{L})$ follows by Lipschitzness and \Cref{lem:lip}.
            The right-hand side with \Cref{eq:sup} concludes the proof of the lemma.
        \end{proof}

        We can combine the above discretization lemma with any (optimal) bandits algorithm to get the desired bound on the regret. For details, we refer to the pseudocode of \banditsmooth. 
        \begin{algorithm}[ht]
        \caption*{\banditsmooth}\label{alg:bandit_smooth}
    \begin{algorithmic}[1]
        \State \textbf{input:} Time horizon $T$, bandit algorithm $\tilde{\cA}$ for gains in $[-1,1]$, grid of $K$ bids $\cX$
        \State Initialize $\tilde{\cA}$ on $K$ actions, one for each $x\in \cX$, time horizon $T$
        \For{each round $t = 1, 2, \dots, T$}
            \State Receive from $\tilde{\cA}$ the bid $B_t \in \cX$
            \State Post bid $B_t$ and observe feedback $Z_t$
            \State Reconstruct $\util_t(B_t)$ from $Z_t$ and feed it to $\tilde\cA$
        \EndFor
    \end{algorithmic}
    \end{algorithm}

        \begin{theorem}
        \label{thm:upper-adv-smooth-bandit}
            Consider the problem of repeated bidding in first-price auctions in an adversarial $\sigma$-smooth environment with bandit feedback. 
            Then there exists a learning algorithm $\cA$ such that 
            \[
            R_T(\cA) \le 
                \tfrac{29}{\sigma} T^{\nicefrac 23}.
            \]
        \end{theorem}
        \begin{proof}
            We prove that algorithm \banditsmooth\ with the right choice of learning algorithm $\tilde{\cA}$ and grid of bids $\cX$ achieves the desired bound on the regret. As learning algorithm $\tilde{\cA}$ we use (a rescaled version of) the Poly INF algorithm \citep{AudibertB10}: since Poly INF is designed to run with gains in $[0,1]$ while the utilities we observe are in $[-1,1]$, we first apply the reward transformation $x \mapsto \frac{x+1}{2}$ to the observed utilities. This transformation costs a multiplicative factor of $2$ in the regret guarantees of Poly INF.
            
            The analysis builds on the discretization result in \Cref{lem:discr_smooth}, by choosing as $\cX$ the uniform grid of $\lceil T^{\nicefrac 23}\rceil + 1$ equally spaced bids on $[0,1]$ (note, $\delta(\cX)$ becomes $T^{-\nicefrac 13}$). Fix any $\sigma$-smooth environment $\cS$, by \Cref{lem:discr_smooth}, the following chain of inequalities holds:
            \begin{align*}
                &\max_{b \in [0,1]}\bbE\lsb{\sum_{t=1}^T \util_t(b)} 
                \le \max_{x \in \cX}\bbE\lsb{\sum_{t=1}^T \util_t(x)} + \frac{6}{\sigma} T^{\nicefrac 23}  \\
                &\qquad \le \bbE\lsb{\sum_{t=1}^T \util_t(B_t)} + \frac{6}{\sigma} T^{\nicefrac 23} + 23 T^{\nicefrac 23}
                \le \frac{29}{\sigma}T^{\nicefrac 23}.
            \end{align*}
            The second inequality follows from the guarantees of (the rescaled version of) Poly INF \citep[Theorem~11]{AudibertB10}.
        \end{proof}
        
    \subsection{Smooth -- Transparent Feedback}
    \label{sec:adv_transparent_smooth}

        For transparent feedback, we combine two tools: the adversarial discretization result (\Cref{lem:discr_smooth}) and the algorithm \expthreefpa{} for learning with transparent feedback on a finite grid. Note, using any other $\sqrt{KT}$ black box learning algorithm (like in the previous section for bandits) would yield a suboptimal regret bound of $T^{\nicefrac 23}$.

        \begin{theorem}
        \label{thm:upper-adv-smooth-transparent}
            Consider the problem of repeated bidding in first-price auctions in an adversarial $\sigma$-smooth environment with transparent feedback. 
            Then there exists a learning algorithm $\cA$ such that
            \[  
                R_T(\cA) \le 6\left(\frac{1}{\sigma} + \sqrt{\ln T}\right) \sqrt{T} \;.   
            \]
        \end{theorem}
        \begin{proof}
            Consider algorithm \expthreefpa{} on the uniform grid $\cX$ of $\lceil\sqrt{T}\rceil+1$ bids, with $\delta(\cX) \le \sqrt{T}$. Fix any $\sigma$-smooth environment $\cS$, \Cref{lem:discr_smooth} implies the following:
            \begin{align*}
                \max_{b \in [0,1]}\bbE\lsb{\sum_{t=1}^T \util_t(b)} &\le \max_{x \in \cX}\bbE\lsb{\sum_{t=1}^T \util_t(x)} + \frac{6}{\sigma} \sqrt T  
                \\
                &\le 
                \bbE\lsb{\sum_{t=1}^T \util_t(B_t)} + 6\left(\frac1 {\sigma} + \sqrt{\ln T}\right)\sqrt T,
            \end{align*}
            where the second inequality follows from \Cref{prop:routine}.
        \end{proof}

    \subsection{The (Non-Smooth) Adversarial Model}\label{sec:adv_full}

        The positive results provided in the previous sections hold under either one of two conditions: the environment is stochastic and the learner has at least the semi-transparent feedback (\Cref{thm:lower-iid-general-bandit} says that bandit feedback is not enough) or the environment uses smooth distributions. These settings allow the learner to compute a discrete class of representative bids efficiently. In this section, we formally argue that learning is impossible if any of these assumptions is dropped. Specifically, the standard adversarial environment that generates the sequence without any smoothness constraint is too strong. In particular, we construct a randomized sequence $(V_1,M_1), (V_2,M_2), \dots$ that induces any learner to suffer at least linear regret. This construction %
        shares some similarities with the lower bound construction in \Cref{thm:lower-iid-general-bandit}, the main difference being that the best bid $b^\star$ is randomized and hidden in such a way that even a learner having access to full feedback cannot pin-point it.

        \begin{theorem}
        \label{thm:lower-adv-general-full}
            Consider the problem of repeated bidding in first-price auctions in an adversarial environment with full feedback. Then, any learning algorithm $\cA$ satisfies $                R_T(\cA)
            \ge
                \nicefrac{T}{24} $.
        \end{theorem}

        \begin{proof}
            We prove the result via Yao's principle, showing that there exists a randomized environment $\cS$ such that any deterministic learning algorithm suffers $\nicefrac{T}{24}$ regret against it. The random sequence posted by $\cS$ is based on two randomized auxiliary sequences $L_1, L_2, \dots$ and $U_1, U_2, \dots$ defined as follows. They are initiated to $L_0 = \nicefrac 12$, $U_0 = \nicefrac 23$. They then evolve recursively as follows: 
            \[
            \begin{cases}
                L_t = L_{t-1} + \tfrac23 {\Delta_{t-1}} \text{ and } U_t = U_{t-1} \text{, with probability } \tfrac 12,\\
                U_t = U_{t-1} - \tfrac{2}3 \Delta_{t-1}\text{ and } L_t = L_{t-1} \text{, with probability } \tfrac 12,
            \end{cases}
            \]
            where $\Delta_{t-1} = U_{t-1} - L_{t-1}$. For each realized sequence of the $(L_t,U_t)$ pairs, the actual sequence of the $(M_t,V_t)$ selected by $\cS$ is constructed as follows. At each time step $t$, the environment selects $(M_t,V_t) = (L_t,1)$ or  $(U_t,0)$, uniformly at random; so that the distribution is characterized by two levels of independent randomness: the auxiliary sequence of shrinking intervals and the choice between $(L_t,1)$ and $(U_t,0)$. 

        We move our attention to the expected performance of the best fixed bid in hindsight. For each realization of the random auxiliary sequence, there exists a bid $B^\star$ such that $(i)$ it wins all the auctions $(V_t,M_t)$ of the form $(L_t,1)$ (which we may call ``good auctions'' because they bring positive utility when won) and $(ii)$ it loses all the auctions $(V_t,M_t)$ of the form $(U_t,0)$ (called ``bad auctions'' because they bring negative utility). Thus its expected utility at each time step is at least $\nicefrac 16$: with probability $\nicefrac 12$ the environment selects a good auction, which induces a utility of $(1-L_t) \ge \nicefrac 13.$ All in all, the optimal bid achieves an expected utility of at least $\nicefrac T6.$

        Consider now the performance of any deterministic algorithm $\cA$: for any fixed time $t>1$ and possible realization of the past observations, the learner posts some deterministic bid $B_t$. If $B_t<L_{t-1}$, then it gets $0$ utility, so we only consider the following cases:
        \begin{itemize}[topsep=2pt,itemsep=0pt,leftmargin=2.5ex,parsep=1pt]
            \item If $B_t \in [L_{t-1},L_{t-1}+\tfrac 13{\Delta_{t-1}})$, then the bidder gets the item with probability $\nicefrac 14$ ($L_{t} = L_{t-1}$, $V_t$ is set to $1$ and $M_t = L_{t}$) with an expected utility of $\nicefrac {(1 - L_t)}4  \le \nicefrac 18$.
            \item If $B_t \in [L_{t-1}+\tfrac 13{\Delta_{t-1}}, L_{t-1}+\tfrac 23{\Delta_{t-1}})$, the bidder gets the item with probability $\nicefrac 12$ (when $L_t = L_{t-1}$ and $U_t = U_{t-1} - \tfrac 23 \Delta_{t-1}$) for an expected utility of $\tfrac 14 (1 - 2L_{t-1} - \tfrac 13 \Delta_{t-1}) \le 0$
            \item If $B_t \in [L_{t-1}+\tfrac 23 \Delta_{t-1}, U_{t-1})$, the bidder gets the item with probability $\nicefrac 34$ (when $L_t = L_{t-1}$ and when $U_t = U_{t-1}$, $V_t = 1$ and $M_t = L_t$) for an expected utility of $\tfrac 14 (1 - L_{t-1}) - \tfrac 14 (L_{t-1} + \tfrac 13 \Delta_{t-1}) + \tfrac14 (1-L_{t-1} - \tfrac 23 \Delta_{t-1}) \le \tfrac 18$ 
            \item If $B_t \ge U_{t-1}$, the bidder always gets the item with a negative expected utility.
        \end{itemize}
        All in all, the expected utility of any deterministic algorithm is at most $\nicefrac {T}{8}$. If we compare this quantity with the lower bound on the expected utility of the best bid in hindsight, we get the desired result: $ \E{R_T(\cA,\cS)} \ge \nicefrac{T}{6} - \nicefrac{T}{8} = \nicefrac{T}{24}.$
        \end{proof}

        A final observation: the main ingredient in the proof is the  elaborate auxiliary sequence. To construct it, we only needed the non-smoothness of $M_t$, while we may have chosen the valuations $V_t$ to be smooth, say uniformly in $[0,\nicefrac 14]$ for the bad auctions and in $[\nicefrac 34,1]$ for the good ones.

    \section{Conclusion}

        Motivated by the recent shift from second to first-price auctions in online advertising markets, this paper comprehensively analyzes the online learning problem of repeated bidding in first-price auctions under the realistic assumption that the bidder does not know its valuation before bidding. We characterize the minimax regret achievable for different levels of transparency in the auction format and different data generation models, considering both the stochastic i.i.d.\ and the standard adversarial model, while also considering smoothness. Although our regret rates are tight in their dependence on the time horizon $T$, a natural open problem is studying their minimax dependence on the smoothness parameter $\sigma$. This paper belongs to the long line of research that studies economic problems from the online learning perspective; an intriguing open problem consists in offering a unified framework to characterize in a satisfying way all these games with partial feedback, similar to what has been done for partial monitoring and feedback graphs.

\section*{Acknowledgment}
NCB, RC, FF, and SL are partially supported by the FAIR (Future Artificial Intelligence Research) project, funded by theNextGenerationEU program within the PNRR-PE-AI scheme (M4C2, investment 1.3, line on Artificial Intelligence). NCB and
RC are also partially supported by the MUR PRIN grant 2022EKNE5K (Learning in Markets and Society) and by the EU Horizon CL4-2022-HUMAN-02 RIA under grant agreement 101120237, project ELIAS (European Lighthouse of AI for Sustainability).
RC also acknowledges the financial support of the Italian Institute of Technology during the writing of this paper. FF and SL are also partially supported by ERC Advanced Grant 788893 AMDROMA and PNRR MUR project IR0000013-SoBigData.it. 

TC gratefully acknowledges the support of the University of Ottawa through grant GR002837 (Start-Up Funds) and that of the Natural Sciences and Engineering Research Council of Canada (NSERC) through grants RGPIN-2023-03688 (Discovery Grants Program) and DGECR-2023-00208 (Discovery Grants Program, DGECR - Discovery Launch Supplement)

\clearpage

\bibliographystyle{plainnat}
\bibliography{references}

\begin{thebibliography}{45}
\providecommand{\natexlab}[1]{#1}
\providecommand{\url}[1]{\texttt{#1}}
\expandafter\ifx\csname urlstyle\endcsname\relax
  \providecommand{\doi}[1]{doi: #1}\else
  \providecommand{\doi}{doi: \begingroup \urlstyle{rm}\Url}\fi

\bibitem[Achddou et~al.(2021)Achddou, Capp{\'{e}}, and Garivier]{AchddouCG21}
Juliette Achddou, Olivier Capp{\'{e}}, and Aur{\'{e}}lien Garivier.
\newblock Fast rate learning in stochastic first price bidding.
\newblock In \emph{{ACML}}, volume 157 of \emph{Proceedings of Machine Learning
  Research}, pages 1754--1769. {PMLR}, 2021.

\bibitem[Alon et~al.(2017)Alon, Cesa{-}Bianchi, Gentile, Mannor, Mansour, and
  Shamir]{AlonCGMMS17}
Noga Alon, Nicol{\`{o}} Cesa{-}Bianchi, Claudio Gentile, Shie Mannor, Yishay
  Mansour, and Ohad Shamir.
\newblock Nonstochastic multi-armed bandits with graph-structured feedback.
\newblock \emph{{SIAM} J. Comput.}, 46\penalty0 (6):\penalty0 1785--1826, 2017.
\newblock \doi{10.1137/140989455}.

\bibitem[Audibert and Bubeck(2009)]{audibert2009minimax}
Jean{-}Yves Audibert and S{\'{e}}bastien Bubeck.
\newblock Minimax policies for adversarial and stochastic bandits.
\newblock In \emph{{COLT}}, 2009.

\bibitem[Audibert and Bubeck(2010)]{AudibertB10}
Jean{-}Yves Audibert and S{\'{e}}bastien Bubeck.
\newblock Regret bounds and minimax policies under partial monitoring.
\newblock \emph{J. Mach. Learn. Res.}, 11:\penalty0 2785--2836, 2010.

\bibitem[Badanidiyuru et~al.(2023)Badanidiyuru, Feng, and
  Guruganesh]{badanidiyuru2023learning}
Ashwinkumar Badanidiyuru, Zhe Feng, and Guru Guruganesh.
\newblock Learning to bid in contextual first price auctions.
\newblock In \emph{{WWW}}, pages 3489--3497. {ACM}, 2023.

\bibitem[Balcan et~al.(2018)Balcan, Dick, and Vitercik]{BalcanDV18}
Maria{-}Florina Balcan, Travis Dick, and Ellen Vitercik.
\newblock Dispersion for data-driven algorithm design, online learning, and
  private optimization.
\newblock In \emph{{FOCS}}, pages 603--614. {IEEE} Computer Society, 2018.
\newblock \doi{10.1109/FOCS.2018.00064}.

\bibitem[Balseiro et~al.(2019)Balseiro, Golrezaei, Mahdian, Mirrokni, and
  Schneider]{balseiro2019contextual}
Santiago~R. Balseiro, Negin Golrezaei, Mohammad Mahdian, Vahab~S. Mirrokni, and
  Jon Schneider.
\newblock Contextual bandits with cross-learning.
\newblock \emph{NeurIPS}, 2019.

\bibitem[Bart{\'{o}}k et~al.(2014)Bart{\'{o}}k, Foster, P{\'{a}}l, Rakhlin, and
  Szepesv{\'{a}}ri]{BartokFPRS14}
G{\'{a}}bor Bart{\'{o}}k, Dean~P. Foster, D{\'{a}}vid P{\'{a}}l, Alexander
  Rakhlin, and Csaba Szepesv{\'{a}}ri.
\newblock Partial monitoring - classification, regret bounds, and algorithms.
\newblock \emph{Math. Oper. Res.}, 39\penalty0 (4):\penalty0 967--997, 2014.
\newblock \doi{10.1287/moor.2014.0663}.

\bibitem[Bass(2013)]{bass2013real}
Richard~F. Bass.
\newblock \emph{Real analysis for graduate students}.
\newblock Createspace Ind Pub, 2013.

\bibitem[Bergemann and H{\"o}rner(2018)]{BergemannH18}
Dirk Bergemann and Johannes H{\"o}rner.
\newblock Should first-price auctions be transparent?
\newblock \emph{American Economic Journal: Microeconomics}, 10\penalty0
  (3):\penalty0 177--218, 2018.

\bibitem[Bernasconi et~al.(2024)Bernasconi, Castiglioni, Celli, and
  Fusco]{BernasconiCCF24}
Martino Bernasconi, Matteo Castiglioni, Andrea Celli, and Federico Fusco.
\newblock No-regret learning in bilateral trade via global budget balance.
\newblock In \emph{{STOC}}. {ACM}, 2024.

\bibitem[Bigler(2019)]{Bigler19}
Jason Bigler.
\newblock Rolling out first price auctions to {G}oogle {A}d {M}anager partners.
\newblock
  \url{https://www.blog.google/products/admanager/rolling-out-first-price-auctions-google-ad-manager-partners/},
  2019.
\newblock Accessed April 7, 2023.

\bibitem[Block et~al.(2022)Block, Dagan, Golowich, and
  Rakhlin]{block2022smoothed}
Adam Block, Yuval Dagan, Noah Golowich, and Alexander Rakhlin.
\newblock Smoothed online learning is as easy as statistical learning.
\newblock In \emph{{COLT}}, volume 178 of \emph{Proceedings of Machine Learning
  Research}, pages 1716--1786. {PMLR}, 2022.

\bibitem[Boli{\'c} et~al.(2024)Boli{\'c}, Cesari, and
  Colomboni]{bolic2023online}
Nata{\v{s}}a Boli{\'c}, Tommaso Cesari, and Roberto Colomboni.
\newblock An online learning theory of brokerage.
\newblock \emph{The 23rd International Conference on Autonomous Agents and
  Multi-Agent Systems}, 2024.

\bibitem[Cesa{-}Bianchi et~al.(2021)Cesa{-}Bianchi, Cesari, Colomboni, Fusco,
  and Leonardi]{cesa21EC}
Nicol{\`{o}} Cesa{-}Bianchi, Tommaso Cesari, Roberto Colomboni, Federico Fusco,
  and Stefano Leonardi.
\newblock A regret analysis of bilateral trade.
\newblock In \emph{{EC}}, pages 289--309. {ACM}, 2021.
\newblock \doi{10.1145/3465456.3467645}.

\bibitem[Cesa{-}Bianchi et~al.(2023)Cesa{-}Bianchi, Cesari, Colomboni, Fusco,
  and Leonardi]{cesa23COLT}
Nicol{\`{o}} Cesa{-}Bianchi, Tommaso Cesari, Roberto Colomboni, Federico Fusco,
  and Stefano Leonardi.
\newblock Repeated bilateral trade against a smoothed adversary.
\newblock In \emph{{COLT}}, volume 195 of \emph{Proceedings of Machine Learning
  Research}, pages 1095--1130. {PMLR}, 2023.

\bibitem[Cesa{-}Bianchi et~al.(2024{\natexlab{a}})Cesa{-}Bianchi, Cesari,
  Colomboni, Fusco, and Leonardi]{Cesa-BianchiCRFL24}
Nicol{\`{o}} Cesa{-}Bianchi, Tommaso Cesari, Roberto Colomboni, Federico Fusco,
  and Stefano Leonardi.
\newblock The role of transparency in repeated first-price auctions with
  unknown valuations.
\newblock In \emph{{STOC}}. {ACM}, 2024{\natexlab{a}}.

\bibitem[Cesa{-}Bianchi et~al.(2024{\natexlab{b}})Cesa{-}Bianchi, Cesari,
  Colomboni, Fusco, and Leonardi]{cesa2024JMLR}
Nicol{\`{o}} Cesa{-}Bianchi, Tommaso Cesari, Roberto Colomboni, Federico Fusco,
  and Stefano Leonardi.
\newblock Regret analysis of bilateral trade with a smoothed adversary.
\newblock \emph{hal preprint hal-04383576}, 2024{\natexlab{b}}.

\bibitem[Cesa{-}Bianchi et~al.(2024{\natexlab{c}})Cesa{-}Bianchi, Cesari,
  Colomboni, Fusco, and Leonardi]{cesa23MOR}
Nicol{\`{o}} Cesa{-}Bianchi, Tommaso Cesari, Roberto Colomboni, Federico Fusco,
  and Stefano Leonardi.
\newblock Bilateral trade: A regret minimization perspective.
\newblock \emph{Mathematics of Operations Research}, 49\penalty0 (1):\penalty0
  171--203, 2024{\natexlab{c}}.
\newblock \doi{10.1287/moor.2023.1351}.

\bibitem[Deng et~al.(2022)Deng, Hu, Lin, and Zheng]{deng2022nash}
Xiaotie Deng, Xinyan Hu, Tao Lin, and Weiqiang Zheng.
\newblock Nash convergence of mean-based learning algorithms in first price
  auctions.
\newblock In \emph{{WWW}}. {ACM}, 2022.

\bibitem[Dikkala and Tardos(2013)]{dikkala2013can}
Nishanth Dikkala and {\'{E}}va Tardos.
\newblock Can credit increase revenue?
\newblock In \emph{{WINE}}, volume 8289 of \emph{Lecture Notes in Computer
  Science}, pages 121--133. Springer, 2013.

\bibitem[Duetting et~al.(2023)Duetting, Guruganesh, Schneider, and
  Wang]{DuettingGSW23}
Paul Duetting, Guru Guruganesh, Jon Schneider, and Joshua~Ruizhi Wang.
\newblock Optimal no-regret learning for one-sided lipschitz functions.
\newblock In \emph{{ICML}}, volume 202 of \emph{Proceedings of Machine Learning
  Research}, pages 8836--8850. {PMLR}, 2023.

\bibitem[Durvasula et~al.(2023)Durvasula, Haghtalab, and
  Zampetakis]{DurvasulaHZ23}
Naveen Durvasula, Nika Haghtalab, and Manolis Zampetakis.
\newblock Smoothed analysis of online non-parametric auctions.
\newblock In \emph{{EC}}, pages 540--560. {ACM}, 2023.

\bibitem[Feldman et~al.(2016)Feldman, Lucier, and Nisan]{feldman2016correlated}
Michal Feldman, Brendan Lucier, and Noam Nisan.
\newblock Correlated and coarse equilibria of single-item auctions.
\newblock In \emph{{WINE}}, volume 10123 of \emph{Lecture Notes in Computer
  Science}, pages 131--144. Springer, 2016.
\newblock \doi{10.1007/978-3-662-54110-4_10}.

\bibitem[Feng et~al.(2018)Feng, Podimata, and Syrgkanis]{FengPS18}
Zhe Feng, Chara Podimata, and Vasilis Syrgkanis.
\newblock Learning to bid without knowing your value.
\newblock In \emph{{EC}}, pages 505--522. {ACM}, 2018.

\bibitem[Feng et~al.(2021)Feng, Guruganesh, Liaw, Mehta, and
  Sethi]{feng2021convergence}
Zhe Feng, Guru Guruganesh, Christopher Liaw, Aranyak Mehta, and Abhishek Sethi.
\newblock Convergence analysis of no-regret bidding algorithms in repeated
  auctions.
\newblock In \emph{{AAAI}}, pages 5399--5406. {AAAI} Press, 2021.
\newblock \doi{10.1609/aaai.v35i6.16680}.

\bibitem[Haghtalab et~al.(2020)Haghtalab, Roughgarden, and
  Shetty]{haghtalab2020smoothed}
Nika Haghtalab, Tim Roughgarden, and Abhishek Shetty.
\newblock Smoothed analysis of online and differentially private learning.
\newblock In \emph{NeurIPS}, 2020.

\bibitem[Haghtalab et~al.(2021)Haghtalab, Roughgarden, and
  Shetty]{HaghtalabRS21}
Nika Haghtalab, Tim Roughgarden, and Abhishek Shetty.
\newblock Smoothed analysis with adaptive adversaries.
\newblock In \emph{{FOCS}}, pages 942--953. {IEEE}, 2021.

\bibitem[Haghtalab et~al.(2022)Haghtalab, Han, Shetty, and
  Yang]{haghtalaboracle}
Nika Haghtalab, Yanjun Han, Abhishek Shetty, and Kunhe Yang.
\newblock Oracle-efficient online learning for smoothed adversaries.
\newblock In \emph{NeurIPS}, 2022.

\bibitem[Han et~al.(2020{\natexlab{a}})Han, Zhou, Flores, Ordentlich, and
  Weissman]{han2020learning}
Yanjun Han, Zhengyuan Zhou, Aaron Flores, Erik Ordentlich, and Tsachy Weissman.
\newblock Learning to bid optimally and efficiently in adversarial first-price
  auctions.
\newblock \emph{arXiv preprint arXiv:2007.04568}, 2020{\natexlab{a}}.

\bibitem[Han et~al.(2020{\natexlab{b}})Han, Zhou, and Weissman]{han2020optimal}
Yanjun Han, Zhengyuan Zhou, and Tsachy Weissman.
\newblock Optimal no-regret learning in repeated first-price auctions.
\newblock \emph{arXiv preprint arXiv:2003.09795}, 2020{\natexlab{b}}.

\bibitem[Kannan et~al.(2018)Kannan, Morgenstern, Roth, Waggoner, and
  Wu]{kannan2018smoothed}
Sampath Kannan, Jamie~H Morgenstern, Aaron Roth, Bo~Waggoner, and Zhiwei~Steven
  Wu.
\newblock A smoothed analysis of the greedy algorithm for the linear contextual
  bandit problem.
\newblock \emph{Advances in neural information processing systems}, 31, 2018.

\bibitem[Kleinberg et~al.(2019)Kleinberg, Slivkins, and Upfal]{KleinbergSU19}
Robert Kleinberg, Aleksandrs Slivkins, and Eli Upfal.
\newblock Bandits and experts in metric spaces.
\newblock \emph{J. {ACM}}, 66\penalty0 (4):\penalty0 30:1--30:77, 2019.
\newblock \doi{10.1145/3299873}.

\bibitem[Kolumbus and Nisan(2022)]{kolumbus2022auctions}
Yoav Kolumbus and Noam Nisan.
\newblock Auctions between regret-minimizing agents.
\newblock In \emph{{WWW}}, pages 100--111. {ACM}, 2022.
\newblock \doi{10.1145/3485447.3512055}.

\bibitem[Lattimore(2022)]{lattimore2022minimax}
Tor Lattimore.
\newblock Minimax regret for partial monitoring: Infinite outcomes and
  rustichini's regret.
\newblock In \emph{{COLT}}, volume 178 of \emph{Proceedings of Machine Learning
  Research}, pages 1547--1575. {PMLR}, 2022.

\bibitem[Mitzenmacher and Upfal(2017)]{MitzenmacherU17}
Michael Mitzenmacher and Eli Upfal.
\newblock \emph{Probability and Computing: Randomized Algorithms and
  Probabilistic Analysis, Second Edition}.
\newblock Cambridge University Press, 2017.
\newblock \doi{10.1017/CBO9780511813603}.

\bibitem[Rakhlin et~al.(2011)Rakhlin, Sridharan, and Tewari]{rakhlin2011online}
Alexander Rakhlin, Karthik Sridharan, and Ambuj Tewari.
\newblock Online learning: Stochastic, constrained, and smoothed adversaries.
\newblock In \emph{{NIPS}}, 2011.

\bibitem[Sharma et~al.(2020)Sharma, Balcan, and Dick]{SharmaBD20}
Dravyansh Sharma, Maria{-}Florina Balcan, and Travis Dick.
\newblock Learning piecewise lipschitz functions in changing environments.
\newblock In \emph{{AISTATS}}, volume 108 of \emph{Proceedings of Machine
  Learning Research}, pages 3567--3577. {PMLR}, 2020.

\bibitem[Slivkins(2019)]{Slivkins19}
Aleksandrs Slivkins.
\newblock Introduction to multi-armed bandits.
\newblock \emph{Found. Trends Mach. Learn.}, 12\penalty0 (1-2):\penalty0
  1--286, 2019.
\newblock \doi{10.1561/2200000068}.

\bibitem[Sluis(2017)]{Sluis17}
Sarah Sluis.
\newblock Big changes coming to auctions, as exchanges roll the dice on
  first-price.
\newblock
  \url{https://adexchanger.com/platforms/big-changes-coming-auctions-exchanges-roll-dice-first-price/},
  2017.
\newblock Accessed July 3, 2023.

\bibitem[Spielman and Teng(2004)]{spielman2004smoothed}
Daniel~A Spielman and Shang-Hua Teng.
\newblock Smoothed analysis of algorithms: Why the simplex algorithm usually
  takes polynomial time.
\newblock \emph{Journal of the ACM (JACM)}, 51\penalty0 (3):\penalty0 385--463,
  2004.
\newblock \doi{10.1145/990308.990310}.

\bibitem[Weed et~al.(2016)Weed, Perchet, and Rigollet]{WeedPR16}
Jonathan Weed, Vianney Perchet, and Philippe Rigollet.
\newblock Online learning in repeated auctions.
\newblock In \emph{{COLT}}, volume~49 of \emph{{JMLR} Workshop and Conference
  Proceedings}, pages 1562--1583. JMLR.org, 2016.

\bibitem[Wong(2021)]{Wong21}
Matt Wong.
\newblock Moving {A}d{S}ense to a first-price auction.
\newblock
  \url{https://blog.google/products/ads-commerce/our-move-to-a-first-price-auction/},
  2021.
\newblock Accessed July 6, 2023.

\bibitem[Zhang et~al.(2021)Zhang, Kitts, Han, Zhou, Mao, He, Pan, Flores,
  Gultekin, and Weissman]{zhang2021meow}
Wei Zhang, Brendan Kitts, Yanjun Han, Zhengyuan Zhou, Tingyu Mao, Hao He,
  Shengjun Pan, Aaron Flores, San Gultekin, and Tsachy Weissman.
\newblock {MEOW:} {A} space-efficient nonparametric bid shading algorithm.
\newblock In \emph{{KDD}}. {ACM}, 2021.

\bibitem[Zhang et~al.(2022)Zhang, Han, Zhou, Flores, and
  Weissman]{zhang2022leveraging}
Wei Zhang, Yanjun Han, Zhengyuan Zhou, Aaron Flores, and Tsachy Weissman.
\newblock Leveraging the hints: Adaptive bidding in repeated first-price
  auctions.
\newblock \emph{NeurIPS}, 2022.

\end{thebibliography}

\clearpage

\appendix

\section{Appendix}

\subsection{Measure and Information-Theoretic Notation and Known Facts}
\label{s:appe-measure}

We recall that given two probability measures $\Pb$ and $\Q$ on a measurable space $(\Omega, \cF)$, $\Q$ is said to be absolutely continuous with respect to $\Pb$ (and we write $\Q \ll \Pb$) if, for all $E\in \cF$ such that $\Pb[E]=0$, it holds that $\Q[E]=0$.
Whenever $\Q \ll \Pb$, the Radon-Nikodym theorem states that there exists a density (called Radon-Nikodym derivative of $\Q$ with respect to $\Pb$) $\frac{\dif \Q}{\dif\Pb}\colon \Omega \to [0,\iop)$ such that, for all $E\in \cF$, it holds that
\[
    \Q[E]
=
    \int_E \frac{\dif \Q}{\dif\Pb}(\omega) \diff \Pb(\omega) \;.
\]
See \cite[Theorem~13.4]{bass2013real} for a reference.

If $(\Omega, \cF, \Pb)$ is a probability space, $(\cX, \cF_{\cX})$ is a measurable space, and $X$ is a random variable from $(\Omega, \cF)$ to $(\cX, \cF_{\cX})$, the push-forward measure of $\Pb$ by $X$ is denoted by $\Pb_X$. In this case, we recall that the push-forward measure is defined as the unique probability measure on $\cF_{\cX}$ defined via $\Pb_X[F] = \Pb[X \in F]$, for all $F \in \cF_{\cX}$.

If $(\Omega, \cF)$ and $(\Omega', \cF')$ are two measurable spaces, their product $\sigma$-algebra is denoted by $\cF \otimes \cF'$. We recall that $\cF \otimes \cF'$ is the $\sigma$-algebra of subsets of $\Omega \times \Omega'$ generated by the collection of subsets of the form $F\times F'$, where $F \in \cF$ and $F' \in \cF'$.
If $(\Omega, \cF, \Pb)$ and $(\Omega', \cF',\Pb')$ are two probability spaces, the product measure of $\Pb$ and $\Pb'$ is denoted by $\Pb\otimes\Pb'$. We recall that $\Pb\otimes\Pb'$ is the unique probability measure defined on $\cF \otimes \cF'$ which satisfies $(\Pb\otimes\Pb')[F\times F'] = \Pb[F]\Pb'[F']$, for all $E\in \cF$ and $E'\in \cF'$.

If $(\Omega, \cF, \Pb)$ is a probability space, $(\cX, \cF_{\cX})$ and $(\cY, \cF_{\cY})$ are measurable spaces, $X$ is a random variable from  $(\Omega, \cF)$ to $(\cX, \cF_{\cX})$, and $Y$ is a random variable from $(\Omega, \cF)$ to $(\cY, \cF_{\cY})$, the conditional probability of $X$ given $Y$ is denoted by $\Pb_{X \mid Y}$, where, for each $E \in\cF_{\cX}$, we recall that $\Pb_{X \mid Y}[E] = \Pb[X \in E \mid Y]$ and that $\Pb_{X \mid Y}[E]$ is a $\sigma(Y)$-measurable random variable.

The following result has been proven in \cite{cesa23COLT}.
\begin{theorem}
\label{t:inverse-transformation-method-2}
Suppose that $(\cY,d)$ is a separable and complete metric space with $\cF_\cY$ as the Borel $\sigma$-algebra of $(\cY,d)$. Let $(\Omega, \cF)$ be a measurable space, $X$ a random variable from $(\Omega, \cF)$ to $\brb{ \{0,1\}, 2^{\{0,1\}} }$,
$Y$ a random variable from $(\Omega,\cF)$ to $(\cY, \cF_\cY)$,
and $U$ random variable from $(\Omega, \cF)$ to $\brb{ [0,1], \cB }$, 
where $\cB$ is the Borel $\sigma$-algebra of $[0,1]$.
Suppose that $\Pb,\Q$ are probability measures defined on $\cF$, and $p \in (0,1)$, $q\in [0,1]$ are such that:
\begin{itemize}[topsep=4pt,itemsep=0pt,leftmargin=9pt]
    \item $\Pb[X=1] = p$ and $\Q[X=1] = q$.
    \item $U$ is a uniform random variable on $[0,1]$ both under $\Pb$ and $\Q$, i.e., we have that $\Pb_U = \leb = \Q_U$.
    \item $U$ is independent of $X$ both under $\Pb$ and $\Q$, i.e., $\Pb_{(X,U)} = \Pb_X \otimes \Pb_U$ and $\Q_{(X,U)} = \Q_X \otimes \Q_U$.
\end{itemize}
Then, the following are equivalent:
\begin{enumerate}
    \item\label{item:inverse-transform-one} There exists a measurable function $\fhi$ from $\brb{ \{0,1\}\times[0,1], 2^{ \{0,1\} } \otimes \cB }$ to $(\cY, \cF_\cY)$ such that
    \[
        \Pb_Y = \Pb_{\fhi(X,U)}
        \qquad \text{ and } \qquad
        \Q_Y = \Q_{\fhi(X,U)} \;.
    \]
    \item \label{item:inverse-transform-two} $\Q_Y \ll \Pb_Y$, and $\Pb_Y$-almost-surely it holds that
    \[
        \min \frac{\dif \Q_X}{\dif \Pb_X}
    \le
        \frac{\dif \Q_Y}{\dif \Pb_Y}
    \le
        \max \frac{\dif \Q_X}{\dif \Pb_X} \;.
    \] 
\end{enumerate}
\end{theorem}

\subsection[Missing Details of the proof of Theorem 3]{Missing Details of the Proof of \Cref{thm:lower-iid-smooth-semi-transparent}}\label{app:lower-iid-smooth-semi-transparent}

    In this section, we will complete the proof of \Cref{thm:lower-iid-smooth-semi-transparent}, showing that the repeated first-price auctions with semi-transparent feedback (in the following, referred to as ``our problem'') are no easier than a $K$-armed bandit instance based on the probability measures $\Pb^1,\dots,\Pb^K$ introduced in \Cref{thm:lower-iid-smooth-semi-transparent}.
    The structure of the proof is inspired by \cite[Section 3]{cesa23COLT}.

    \paragraph{The related bandit problem.}
    The action space is $[K]$, where we recall that $K$ was some arbitrarily fixed natural number.
    Let $Y,Y_1,Y_2,\dots$ be a sequence of $\{0,1\}^K$-valued random variables such that, for any $k\in\{0,1,\dots,K\}$, the sequence is $\Pb^k$-i.i.d.\ and, for all $j\in[K]$
    \[
        \Pb^k\bsb{Y(j)=1}
    =
        \begin{cases}
            1/2 & \text{ if } j \neq k\\
            1/2 + 1/(6K) & \text{ if } j = k
        \end{cases}
    \]
    This sequence of latent random variables will determine the rewards of the actions.
    The reward function is
    \[
        \rho \colon [K] \times \{0,1\} \to [0,1]\;, \qquad (i,y) \mapsto \frac{23 + 2 y(i)}{192} 
    \]
    and the feedback received after playing an action $I_t$ at time $t$ is $Y_t (I_t)$ (which is equivalent to receiving the bandit feedback $\rho(I_t, Y_t)$ gathered at time $t$).
    
    For any $k\in\{0,\dots,K\}$ and any $i\in[K]$ the expected reward is
    \[
        \bbE^k\bsb{ \rho(i,Y) }
    =
        \begin{cases}
            \displaystyle \frac{1}{8} & \text{ if } i \neq k \vspace{1ex}\\ 
            \displaystyle \frac{1}{8} + \frac{\e}{144} & \text{ if } i = k
        \end{cases}
    \]

    \paragraph{Mapping our problem into this bandit problem.}
    Assume that $K\ge 3$.
    We partition the interval $[0,1]$ in the following $K$ disjoint regions: $ J_1 = [0, w_1 + \e) $, $J_k = [w_k - \e, w_k + \e)$ (for all $k \in \{2, \dots, K-1\}$), and $ J_K = [w_{K} - \e, 1] $.
    We define a function $\iota \colon [0,1] \to [K]$ that maps each point in the interval $[0,1]$ to one of the $K$ arms by mapping each $b \in [0,1]$ to the unique $i \in [K]$ such that $b \in J_{i}$ (for a pictorial representation of the map $\iota$, see \Cref{fig:mapping}).
    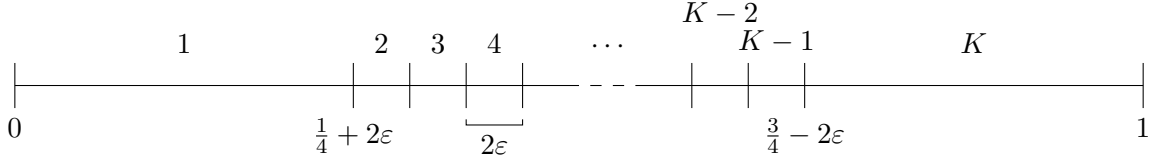
\begin{figure}
        \centering
        \def\lunghezza{20}
        \begin{tikzpicture}[scale = 15]
                    \draw (0,0) -- ({10/\lunghezza},0);
                    \draw ({11/\lunghezza},0) -- (1,0);
                    \draw ({10.2/\lunghezza},0) -- ({10.4/\lunghezza},0);
                    \draw ({10.6/\lunghezza},0) -- ({10.8/\lunghezza},0);
                    \foreach \i in {0,6,7,8,9,12,13,14,20} {
                        \draw ({0+\i/\lunghezza},0.02) -- ({0+\i/\lunghezza},-0.02);
                    }
                    \draw (0 ,-0.02) node[below] {$0$}
                        (6/\lunghezza ,-0.02) node[below] {$\frac14+2\e$}
                        (14/\lunghezza ,-0.02) node[below] {$\frac34-2\e$}
                        (1 ,-0.02) node[below] {$1$}
                        (3/\lunghezza,0.02) node[above] {$1$}
                        ({7/\lunghezza - 1/(2*\lunghezza)},0.02) node[above] {$2$}
                        ({8/\lunghezza - 1/(2*\lunghezza)},0.02) node[above] {$3$}
                        ({9/\lunghezza - 1/(2*\lunghezza)},0.02) node[above] {$4$}
                        ({10.55/\lunghezza},0.02) node[above] {$\cdots$}
                        ({13/\lunghezza - 1/(2*\lunghezza)},0.045) node[above] {$K-2$}
                        ({14/\lunghezza - 1/(2*\lunghezza)},0.02) node[above] {$K-1$}
                        ({17/\lunghezza},0.02) node[above] {$K$}
                        ;
                    \draw (8/\lunghezza,-0.03) -- (8/\lunghezza,-0.035) -- (9/\lunghezza,-0.035) -- (9/\lunghezza,-0.03);
                    \draw (8.5/\lunghezza, -0.035) node[below] {$2 \e$};
                \end{tikzpicture}
        \caption{\footnotesize  A representation of the map $\iota$ through which the bids in the first-price auction problem are related to the $K$-arms of the bandit problem. The interval $[0,1]$ is partitioned in $K$ disjoint intervals, the first and the last one of length $\nicefrac{1}{4}+2\e$, and all the ones in between of length $2\e$. $\iota$ maps each bid to the index of the interval to which it belongs. }
        \label{fig:mapping}
    \end{figure}

    \paragraph{Simulating the feedback.}
    To lighten the notation, besides the already defined random functions $\psi_1,\psi_2,\dots$, define also:
    \[
        \psi \colon [0,1] \to \brb{ [0,1] \times \{\star\} } \cup \brb{ \{\star\} \times [0,1] } \;,
    \qquad b \mapsto
        \begin{cases}
            (V,\star) & \text{ if } b \ge M \\
            (\star, M) & \text{ if } b < M
        \end{cases}
    \]
    The next lemma shows that we can use the feedback observed in the bandit problem together with some independent noise to simulate exactly the feedback of our problem.
    \begin{lemma}
        \label{lemma:pappappero}
        For each $b\in[0,1]$, there exists $\fhi_b \colon \{0,1\} \times [0,1] \to \domain$ such that, if $U'$ is a $[0,1]$-valued random variable such that, for each $k\in\{0,\dots,K\}$, the distribution $U'$ with respect to $\Pb^k$ is a uniform on $[0,1]$ and $U'$ is $\Pb^k$-independent of $Y$, then $\Pb^k_{\fhi_b(Y(\iota(b)),U')} = \Pb^k_{ \psi(b) }$.
    \end{lemma}
    \begin{proof}[Proof of \Cref{lemma:pappappero}]
        A direct verification shows that, for all $k\in[K]$ and all $b\in[0,1]$, 
        $\Pb^k_{\psi(b)} \ll \Pb^0_{\psi(b)}$ (i.e., $\Pb^k_{\psi(b)}$ is absolutely continuous with respect to $\Pb^0_{\psi(b)}$)
        and the Radon-Nikodym derivative of the push-forward measure $\Pb^k_{\psi(b)}$ with respect to $\Pb^0_{\psi(b)}$ satisfies, 
        for $\Pb^0_{\psi(b)}$-a.e.~$(v,m) \in \brb{ [0,1] \times \{\star\} } \cup \brb{ \{\star\} \times [0,1] }$,
        \[
            \frac{\mathrm{d}\Pb^k_{\psi(b)}}{\mathrm{d}\Pb^0_{\psi(b)}} (v,m)
        =
            1 + \e \cdot \frac{16}{9} \lrb{ v-b } \mathrm{sgn}\lrb{ v - \frac{15}{16} } \Lambda_{w_k,\e} (b) \I \lcb{ v \in \lsb{ \frac{7}{8}, 1 } }
        \]
        which implies, for $\Pb^0_{\psi(b)}$-a.e.~$(v,m) \in \brb{ [0,1] \times \{\star\} } \cup \brb{ \{\star\} \times [0,1] }$, that
        \[
            \min\lrb{ \frac{\mathrm{d}\Pb^k_{Y(\iota(b))}}{\mathrm{d}\Pb^0_{Y(\iota(b))}} } 
        =  
            1-\frac{4}{3}\e 
        \le 
            \frac{\mathrm{d}\Pb^k_{\psi(b)}}{\mathrm{d}\Pb^0_{\psi(b)}} (v,m) 
        \le 
            1+\frac{4}{3} \e 
        = 
            \max\lrb{ \frac{\mathrm{d}\Pb^k_{Y(\iota(b))}}{\mathrm{d}\Pb^0_{Y(\iota(b))}} }
        \] 
        Thus, for each $b \in [0,1]$, by \cref{t:inverse-transformation-method-2}, there exists (and we fix)
        \[
            \fhi_{b} \colon \{0,1\} \times [0,1] \to \domain
        \]
        such that
        \[
            \Pb^{\iota(b)}_{\fhi_{b}(Y(\iota(b)),U')}
        =
            \Pb^{\iota(b)}_{\psi(b)}
        \qquad \text{ and } \qquad
            \Pb^{0}_{\fhi_{b}(Y(\iota(b)),U')}
        =
            \Pb^{0}_{\psi(b)} \;.
        \]
        Since for all $b\in [0,1]$ and all $k \in [K]\m\bcb{ \iota(b) }$, we have  
        $
            \Pb^{k}_{\psi(b)}
        =
            \Pb^{0}_{\psi(b)}
        $ (by \Cref{eq:same-stuff}) and
        $
            \Pb^{k}_{\fhi_{b}(Y(\iota(b)),U')}
        =
            \Pb^{0}_{\fhi_{b}(Y(\iota(b)),U')}
        $, then, for all $b \in[0,1]$ and all $k\in\{0,\dots,K\}$, it holds that
        \[
            \Pb^{k}_{\fhi_{b}(Y(\iota(b)),U')}
        =
            \Pb^{k}_{\psi(b)} \;. \qedhere
        \]
    \end{proof}

    We now show that any algorithm $\cA$ for our problem can be transformed into an algorithm $\widetilde\cA$ to solve the bandit problem that suffers no-larger regret. 
    To do so, we begin by formally explaining how algorithms for our problem work.

    \paragraph{Functioning of an algorithm $\cA$ for our problem}
    A randomized algorithm $\cA$ for our problem is a sequence of functions that take as input a sequence of random seeds $U_1, U_2,\dots$ and some feedback $Z_1, Z_2, \dots$ and generates bids $B_t$ as described below. 
    At time $t=1$, $\cA$ selects a bid $B_1$ as a deterministic function of $U_1$ and observes feedback $Z_1 = \psi_1(B_1)$.
    Inductively, for any $t\ge 2$, $\cA$ selects a bid $B_t$ as a deterministic function of $U_1,\dots,U_t,Z_1,\dots,Z_{t-1}$ (where $Z_s = \psi_s(B_s)$, for all $s\in[t-1]$).
    For all $k\in\{0,\dots,K\}$, the sequence of seeds is a $\Pb^k$-i.i.d.\ sequence of uniform random variables on $[0,1]$ that is $\Pb^k$-independent of $(V,M), (V_1,M_1),(V_2,M_2), \dots$.

    \paragraph{Building $\cAt$ from $\cA$}
    We show now how to map $\cA$ to an algorithm $\cAt$ (that shares the same seeds for the randomization) for the bandit problem that suffers a worst-case regret that is no larger than that of $\cA$.

    To do so, consider a sequence $U',U'_1,\dots$ of random variables that, for all $k\in\{0,\dots,K\}$ is a $\Pb^k$-i.i.d.\ sequence of uniforms on $[0,1]$ that $\cAt$ can access as a further source of randomness.
    We will assume that, for all $k\in\{0,\dots,K\}$, the four sequences $Y,Y_1,\dots$, $(V,M),(V_1,M_1),\dots$, $U,U_1,\dots$,and $U',U'_1,\dots$ are independent of each other.

    The algorithm $\cAt$ acts as follows.
    At time $1$, $\cAt$ plays the arm $\tI_1 = \iota(B'_t)$, where $B'_1 = B_1$ is the bid played by $\cA$ at round $t=1$ (chosen as a deterministic function of the random seed $U_1$).
    Then $\cAt$ observes the bandit feedback $Y_1(\tI_1)$ and feeds back to $\cA$ the surrogate feedback $Z'_1 = \fhi_{B'_1} \brb{ Y_1(\tI_1), U'_1 }$.
    Then, inductively, for any time $t\ge 2$, assuming that $\cAt$ played arms $\tI_1,\dots,\tI_{t-1}$ and fed back to $\cA$ the surrogate feedback $Z'_1, \dots, Z'_{t-1}$, then
    \begin{enumerate}
        \item $\cAt$ plays the arm $\tI_t = \iota(B'_t)$, where $B'_t$ is the bid played by $\cA$ at round $t$ (chosen as a deterministic function of the random seeds $U_1,\dots,U_{t}$ and past surrogate feedback $Z'_1,\dots,Z'_{t-1}$).
        \item  $\cAt$ observes the bandit feedback $Y_t(\tI_t)$ and feeds back to $\cA$ the surrogate feedback $Z'_t = \fhi_{B'_t} \brb{ Y_t(\tI_t), U'_t }$.
    \end{enumerate}
    This way, we defined by induction the randomized algorithm $\cAt$.

    By induction on $t$, one can show that, if $B_1,B_2,\dots$ are the bids played by $\cA$ on the basis of the feedback $Z_1 = \psi_1(B_1), Z_2 = \psi_2(B_2), \dots$, then, for all $k\in\{0,\dots,K\}$, we have
    \[
        \Pb^k_{(B_t,Y_t)}
    =
        \Pb^k_{(B_t',Y_t)}
    \]
    which leads to
    \begin{align*}
    &
        R_T^k(\cA)
    =
        T\cdot \bbE^k\bsb{ \util(w_k) } - \sum_{t=1}^T \bbE^k\bsb{ \util_t(B_t) }
    \ge
        T\cdot \bbE^k\bsb{ \rho(k,Y) } - \sum_{t=1}^T \bbE^k\Bsb{ \rho\brb{ \iota(B_t),Y_t } }
    \\ 
    &
    \phantom{R_T^k(\cA) }
    =
        T\cdot \bbE^k\bsb{ \rho(k,Y) } - \sum_{t=1}^T \bbE^k\Bsb{ \rho\brb{ \iota(B'_t),Y_t } }
    =
        T\cdot \bbE^k\bsb{ \rho(k,Y) } - \sum_{t=1}^T \bbE^k\Bsb{ \rho\brb{ \tI_t, Y_t } }
    =
        \widetilde{R}^k_T(\cAt)
    \end{align*}
    (the last equality is a definition).
    Now we are left to show only that for any algorithm $\hat{\cA}$ for the bandit problem which plays actions $I_1,I_2,\dots$, there exists $k\in[K]$ such that
    \[
        \widetilde R_T^k( \hat{\cA} )
    =
        T\cdot \bbE^k\bsb{ \rho(k,Y) } - \sum_{t=1}^T \bbE^k\Bsb{ \rho\brb{ I_t, Y_t } }
    =
        \Omega( T^{\nicefrac 23} )
    \]
    (the first equality is a definition).
    By Yao's Minimax principle, it is sufficient to show this for deterministic algorithms $\hat{\cA}$ for the bandit problem.
    \begin{lemma}
\label{lemma:finale}
    Fix any deterministic algorithm $\hat{\cA}$ for the bandit problem on $K$ actions, then there exists $k\in[K]$ such that $\widetilde R_T^k  (\hat{\cA}) \ge \frac{3}{10^4} T^{\nicefrac 23} $.
    \end{lemma}
    \begin{proof}
    For any deterministic algorithm $\hat{\cA}$ for the bandit problem on $K$ actions, let $I_1, I_2, \dots$ be the actions played by $\hat{\cA}$ on the basis of the sequential feedback received $Z_1, Z_2, \dots$ and define $N_t(i)$ as the random variables counting the number of times the learning algorithm $\hat{\cA}$ plays action $i$, up to time $t$, for any $i \in [K]$ and any time $t \in [T]$:
    \[
        N_t(i) = \sum_{s=1}^t \I\{I_s = i\} .
    \]
    We relate the expected values of $N_T(k)$ under $\Pb^0$ and $\Pb^k$ as a function of the expected number of times the algorithm plays the corresponding actions $k$. This formalizes the intuition that to discriminate between the different $\Pb^k$ the learner needs to play exploring actions. 
    \begin{claim}
    \label{claim:finale}
        The following inequality holds true for any $k\in[K]$:
        \begin{equation}
            \bbE^k \bsb{ N_T(k) } - \bbE^0 \bsb{ N_T(k) } \le  \frac23 \cdot \e \cdot T \cdot \sqrt{ 2 \bbE^0 [N_T(k)] }.    
        \end{equation}
    \end{claim}
    \begin{proof}[Proof of \Cref{claim:finale}]
        For any $t\in [T]$, the action $I_t = I_t(Z_1, \dots, Z_{t-1})$ selected by $\hat{\cA}$ at round $t$ is a deterministic function of $Z_1, \dots, Z_{t-1}$, for each $k\in[K]$. In formula, we then have the following
    \begin{align}
    \nonumber
     \bbE^k \bsb{ N_T(k) } - \bbE^0 \bsb{ N_T(k) } &=
        \sum_{t = 2}^T \Brb{ \Pb^k\bsb{ I_t (Z_1, \dots, Z_{t-1} ) = k } - \Pb^0\bsb{ I_t (Z_1, \dots, Z_{t-1} ) = k } }\\
    \label{eq:TV}&\le
        \sum_{t = 2}^T \bno{ \Pb^k_{(Z_1, \dots, Z_{t-1} )} - \Pb_{(Z_1, \dots, Z_{t-1} )}^0 }_{\mathrm{TV}},
    \end{align}
    where $\lno{\cdot}_{\mathrm{TV}}$ denotes the total variation norm. We move now our attention towards bounding the total variation norm. To that end we use Pinsker's inequality and apply the chain rule for the KL divergence $\kl$. For each $k \in [K]$ and $t \in [T]$ we have the following:
    
    \begin{align}
    \nonumber
        \bno{ \Pb_{(Z_1,\dots,Z_t)}^0 &- \Pb^k_{(Z_1,\dots,Z_t)}}_{\mathrm{TV}}
    \le 
        \sqrt{\frac12 \kl \brb{ \Pb_{(Z_1,\dots,Z_t)}^0,\, \Pb^k_{(Z_1,\dots,Z_t)} }}
    \\
    \label{eq:secondKL}
    &\le
        \sqrt{ \frac12 \lrb{ \kl\brb{ \Pb_{Z_1}^0, \, \Pb^k_{Z_1} } + \sum_{s=2}^t \bbE\Bsb{ \kl \brb{ \Pb_{Z_s \mid Z_1,\dots,Z_{s-1} }^0 , \, \Pb^k_{Z_s \mid Z_1,\dots,Z_{s-1}} } } } }
    \end{align}
    We bound the two KL terms separately. $\hat{\cA}$ is a deterministic algorithm, thus $I_1$ is a fixed element of $[K]$, which implies that, for all $k\in [K]$,
    \begin{align}
    \nonumber
    &
        \kl\brb{ \Pb_{Z_1}^0, \, \Pb^k_{Z_1} }
    \\
    \nonumber 
    & \quad
    =
        \lrb{ 
        \ln\lrb{ \frac{\Pb^0[Y_1(k) = 0]}{ \Pb^k[Y_1(k) = 0] } } \Pb^0[Y_1(k) = 0]
        +
        \ln\lrb{ \frac{\Pb^0[Y_1(k) = 1]}{ \Pb^k[Y_1(k) = 1] } } \Pb^0[Y_1(k) = 1]
        } \I\bcb{ I_1 = k }
    \\
    \label{eq:KLterm1}
    & \quad
    =
        \frac12\lrb{
        \ln \frac{\nicefrac{1}{2}}{\nicefrac{1}{2} - \cprob \cdot \e}
        +
        \ln\frac{\nicefrac{1}{2}}{\nicefrac{1}{2} + \cprob \cdot \e}
        } \cdot \I\{I_1 = k\}
    \end{align}
    Similarly, since $\hat{\cA}$ is a deterministic algorithm, for all $s \ge 2$, the action $I_s = I_s(Z_1,\dots,Z_{s-1})$ selected by $\hat{\cA}$ at time $t$ is a function of $Z_1,\dots, Z_{s-1}$ only, which implies, for all $k \in [K]$, 
    \begin{align}
    \nonumber
    &
        \kl \brb{ \Pb_{Z_s \mid Z_1,\dots,Z_{s-1} }^0 , \, \Pb^k_{Z_s \mid Z_1,\dots,Z_{s-1}} }
    \\
    \nonumber
    & 
    \qquad =
        \bbE^0 \left[
            \ln \lrb{ \frac{ \Pb^0[Z_s = 0 \mid Z_1, \dots, Z_{s-1} ] }{ \Pb^k[Z_s = 0 \mid Z_1, \dots, Z_{s-1} ] } } \Pb^0[Z_s = 0 \mid Z_1, \dots, Z_{s-1} ]
            \right.
    \\
    \nonumber
    & \qquad 
            \qquad
            +
            \left.
            \ln \lrb{ \frac{ \Pb^0[Z_s = 1 \mid Z_1, \dots, Z_{s-1} ] }{ \Pb^k[Z_s = 1 \mid Z_1, \dots, Z_{s-1} ] } } \Pb^0[Z_s = 1 \mid Z_1, \dots, Z_{s-1} ]
        \right]
    \\
    \nonumber
    &
    \qquad
    =
        \bbE^0 \bigg[ \lrb{ 
        \ln\lrb{ \frac{\Pb^0[Y_s(k) = 0]}{ \Pb^k[Y_s(k) = 0] } } \Pb^0[Y_s(k) = 0]
        +
        \ln\lrb{ \frac{\Pb^0[Y_s(k) = 1]}{ \Pb^k[Y_s(k) = 1] } } \Pb^0[Y_s(k) = 1]
        } 
    \\
    \nonumber& \qquad \qquad 
        \times \I\bcb{ I_s (Z_1, \dots, Z_{s-1}) = k } \bigg]
    \\
    \label{eq:KLterm2}
    &
    \qquad
    =
        \frac12\lrb{
        \ln \frac{\nicefrac{1}{2}}{\nicefrac{1}{2} - \cprob \cdot \e}
        +
        \ln\frac{\nicefrac{1}{2}}{\nicefrac{1}{2} + \cprob \cdot \e}
        }
        \Pb^0 \bsb{ I_s (Z_1, \dots, Z_{s-1}) = k } 
    \end{align}
    Now, since $\e = \frac{1}{4K} \le \frac{1}{4} \le \frac{2}{3}$, the following useful inequality holds:
    \begin{equation}
    \label{eq:useful}
        \frac12\lrb{
        \ln \frac{\nicefrac{1}{2}}{\nicefrac{1}{2} - \cprob \cdot \e}
        +
        \ln\frac{\nicefrac{1}{2}}{\nicefrac{1}{2} + \cprob \cdot \e}
        }
    \le
        4 \cdot \lrb{\cprob}^2 \cdot \e^2.
    \end{equation}
    We can combine the inequalities in \Cref{eq:KLterm1} and \Cref{eq:KLterm2} into \Cref{eq:secondKL} and plug in the bound in to obtain:
    \[
        \lno{ \Pb_{(Z_1,\dots,Z_t)}^0 - \Pb^k_{(Z_1,\dots,Z_t)}}_{\mathrm{TV}} \le \cprob \cdot \e \cdot \sqrt{ 2 \bbE [N_t(k)] }
    \]
    Once we have this upper bound on the total variations of the random variables $(Z_1,\dots,Z_t)$ under $\Pb^0$ and $\Pb^k$ we can get back to the initial \Cref{eq:TV} and obtain the desired bound via Jensen: 
    
    \[
        \bbE^k \bsb{ N_T(k) } - \bbE^0 \bsb{ N_T(k) }
    \le
        \sum_{t=2}^T \cprob \cdot \e \cdot \sqrt{ 2 \bbE^0 [N_{t-1}(k)] }
    \le
        \cprob \cdot \e \cdot T \cdot \sqrt{ 2 \bbE^0 [N_T(k)] }. \qedhere
    \]    
    \end{proof}
    Averaging the quantitative bounds in \Cref{claim:finale} for all $k$ in $[K]$, and applying Jensen's inequality, we get the following:
        \begin{align}
        \nonumber
        \frac{1}{K} \sum_{k \in [K]} \bbE^k[N_T(k)]
    &\le
        \frac{1}{K} \sum_{k \in [K]} \bbE^0[N_T(k)] + \cprob \cdot \e \cdot T \cdot \sqrt{ \frac{2}{K} \sum_{k \in [K]}\bbE^0 \lsb{  N_T(k) } }
    \\
    &=
        \lrb{ \frac{1}{K} +  \cprob \cdot \e \cdot \sqrt{ \frac{2T}{K} } } \cdot T \;.
    \end{align}
    Now, we have all the ingredients to lower bound the average regret suffered by $\hat{\cA}$. Note that every time a suboptimal arm is played the learner suffers (expected) instantaneous regret equal $\frac{1}{144} \cdot \e$.
    Then, recalling that $\e=1/(4K)$ and setting $K = \bce{ T ^{1/3} }$ we have, for all $T\ge 8$,
        \begin{align*}
            \frac{1}{K} \sum_{k\in [K]} \tilde R_T^k (\hat{\cA})
        &=
            \frac{1}{K} \sum_{k\in [K]} \Brb{ \cspike \cdot \e \cdot \bbE^k\bsb{ T - N_T(k) }  }
        =
            \cspike \cdot \e \lrb{ T - \frac{1}{K} \sum_{k\in [K]} \bbE^k\bsb{ N_T(k) }} 
        \\
        &
        \ge
            \cspike \cdot \e \cdot  \lrb{ 1 - \frac{1}{K} - \cprob \cdot \e \cdot \sqrt{ \frac{2T}{K} } } \cdot T
        =
            \cspike \cdot \frac{1}{4K} \cdot  \lrb{ 1 - \frac{1}{K} - \frac{1}{6K} \cdot \sqrt{ \frac{2T}{K} } } \cdot T
        \\
        &
        \ge
            \frac{1}{8 \cdot 144}\Brb{\frac{3-\sqrt{2}}{6}}T^{\nicefrac 23}  \ge \frac{3}{10^4} T^{\nicefrac 23}\;.
    \end{align*}
    Therefore, for all $T\ge 8$, there exists $k\in[K]$ such that $\tilde R_T^k (\hat{\cA}) \ge (\fracc{3}{10^4}) \cdot T^{\nicefrac 23}$, concluding the proof.
    \end{proof}

    \subsection[Missing Details of the proof of Theorem 6]{Missing Details of the Proof of \Cref{thm:lower-iid-smooth-full}} \label{app:lower-iid-smooth-full}

        \suboptimality*
            \begin{proof}
                For any $\e \in (0,\tfrac 12)$, the distributions $\Pb^{\pm \e}$ are such that, the set of all the bids that induce non-negative utility $\mathbb{E}^{\pm \e}[\util_t(b)]$ is contained into two disjoint intervals $I_+ = [0,\tfrac 18]$ and $I_- = [\tfrac 14,1]$\footnote{The choice of $I_+$ and $I_-$ is not tight.}. 
                
                We consider separately the two cases $\Pb^{+\e}$ and $\Pb^{-\e}$. We start from the former. By simply looking at the definition (\ref{eq:Epm}), it is clear that $\mathbb E ^{+ \e} [\util_t(b)]$ is monotonically increasing in $\e$ for any $b \in I_+$, on the contrary, it is monotonically decreasing for $b \in I_-$. We have the following:
                \begin{align*}
                    \max_{b \in I_-}\mathbb E ^{+ \e} [\util_t( b)] \le \max_{b \in I_-}\mathbb E ^{0} [\util_t( b)]= \tfrac 1{128}. 
                \end{align*}
                On the other hand, 
                \[
                    \max_{x \in [0,1]}\mathbb E ^{+ \e} [\util_t(x)] \ge \mathbb E ^{+ \e} [\util_t(\tfrac 1{16})] = \tfrac 1{128}(1+\e)> \max_{b \in I_-}\mathbb E ^{+ \e} [\util_t( b)] + \tfrac \e{128}.
                \]

                We consider now the other case, corresponding to $\Pb^{-\e}$. By the definition in \Cref{eq:Epm}, $\mathbb E ^{- \e} [\util_t(b)]$ is monotonically increasing in its first argument for any $b \in I_-$, on the contrary, it is monotonically decreasing for $b \in I_+$. Similarly to the other case we have two steps. On the one hand, it holds that
                \begin{align*}
                    \max_{b \in I_+}\mathbb E ^{- \e} [\util_t( b)] \le \max_{b \in I_+}\mathbb E ^{0} [\util_t( b)] = \tfrac 1{128},
                \end{align*}
                while on the other hand it holds that
                \[
                    \max_{x \in [0,1]}\mathbb E ^{- \e} [\util_t(x)] \ge \mathbb E ^{- \e} [\util_t(\tfrac 7{16})] = \tfrac 1{128} + \e\tfrac{41}{128} > \max_{b \in I_-}\mathbb E ^{+ \e} [\util_t( b)] + \tfrac \e4. \qedhere
                \]
            \end{proof}

             We need a preliminary result for the proof of \Cref{cl:last_step}. Recall, we use the same random variable $(V,M)$ to denote the highest competing bid/valuation pair drawn from the different probability distribution. When we change the underlying measure, we are changing its law. Consider now the push forward measures on $[0,1]^2$ (with the Borel $\sigma$-algebra) induced by these three measures: $\mathbb P_{(V,M)}^0$, $\mathbb P_{(V,M)}^{+\e}$ and $\mathbb P_{(V,M)}^{-\e}$. With some simple calculations (similarly to what is done in, e.g., Appendix B of \cite{Slivkins19}) it is possible to bound the KL divergence:
            \begin{claim}\label{eq:KLaux}
            For any $\e\in (0,\tfrac 12)$ the following inequality holds true:
                \[
                    \kl\left(\mathbb P_{(V,M)}^{+\e},\mathbb P_{(V,M)}^{0}\right) = \kl\left(\mathbb P_{(V,M)}^{-\e}, \mathbb P_{(V,M)}^{0}\right) \le 2 \eps^2
                \]    
            \end{claim}
            \begin{proof}
            We simply apply the definition of $\kl$ divergence for continuous random variables. We only do the calculations for $\mathbb P_{(V,M)}^{+\e}$, the other term is analogous:
            \begin{align*}
                \kl\left(\mathbb P_{(V,M)}^{+\e},\mathbb P_{(V,M)}^{0}\right) &=  
                \int_{Q_+ \cup Q_-} f^{+\e}(v,m) \ln \frac{f^{+\e}(v,m)}{f^0(v,m)} dm \, dv \\
                &=   \frac 12 (1+\e) \ln (1+\e) + \frac 12 (1-\e) \ln (1-\e) 
                \le 2\e^2,
            \end{align*}
            where the last inequality holds for any $\e \in (0,\frac 12)$.
            \end{proof}

            \lastStep*
            \begin{proof}  We have the following:
            \begin{align}
            \nonumber
                \Ei{N_i} - \Eo{N_i} &= \sum_{t=2}^T \Pb^i\lsb{B_t \in I_i} - \Pb^0 \lsb{B_t \in I_i}\\
                &\le\sum_{t=2}^T ||\mathbb P^i_{(V_1,M_1), \dots,(V_{t-1},M_{t-1}) } - \mathbb P^0_{(V_1,M_1), \dots,(V_{t-1},M_{t-1}) }||_{\tv} \tag{\text{Total variation}}\\
                &\le\sum_{t=2}^T \sqrt{\frac 12 \kl\left(\mathbb P^i_{(V_1,M_1), \dots,(V_{t-1},M_{t-1})},\mathbb P^0_{(V_1,M_1), \dots,(V_{t-1},M_{t-1})}\right)} \tag{\text{Pinsker's inequality}}\\
                &\le\sum_{t=2}^T \sqrt{\frac t2 \kl\left(\mathbb P^i_{(V,M)},\mathbb P^0_{(V,M)}\right)}\tag{\text{$(V_1,M_1), \dots,(V_{t-1},M_{t-1}), \dots$ are i.i.d.}}\\
                \label{eq:KL}
                &\le\frac{1}{4\sqrt{T}}\sum_{t=2}^T \sqrt{t} \le \frac 14 T,
            \end{align}
            where in the last inequality we applied \Cref{eq:KLaux} for our choice of $\e = 1/(4\sqrt{T})$. Note, $\Pb^j_{(V_1,M_1), \dots,(M_{t},V_{t})}$ is the push-forward measure on $([0,1]^2)^{t}$ induced by $t$ i.i.d.\ draws of $(V,M)$ from distribution $\mathbb{P}^j$, $j \in \{0,1,2\}$.  
            Averaging the result in \Cref{eq:KL}, we get the desired inequality:
            \[
                \frac 12 \sum_{i=1,2}\Ei{N_i} \le \frac 12 \sum_{i=1,2}\Eo{N_i} + \frac T4 = \frac 34 T. \qedhere
            \]
            \end{proof}

\end{document}